\pdfoutput=1

\documentclass[12pt]{article}
\usepackage{amsmath}
\usepackage{graphicx,psfrag,epsf}
\usepackage{enumerate}
\usepackage{natbib}
\usepackage{url} 

\newcommand{\blind}{1}

\usepackage[left=1.1in,top=1.40in,right=1.1in,bottom=1.6in,head=1.40in]{geometry}

\usepackage[T1]{fontenc}
\usepackage[utf8]{inputenc}
\usepackage{amsmath,amssymb}

\usepackage{amsmath}
\usepackage{amssymb}
\usepackage{graphicx}
\usepackage[toc,page]{appendix}
\usepackage{hyperref}
\usepackage{subcaption}
\usepackage{amsthm}
\usepackage{enumitem}
\usepackage{color}
\usepackage{colortbl}
\usepackage{multirow}
\usepackage{float}
\usepackage{natbib}
\usepackage[english]{babel}

\setcitestyle{aysep={}}

\newtheorem{theorem}{Theorem}
\newtheorem{lemma}[theorem]{Lemma}
\newtheorem{corollary}[theorem]{Corollary}
\newtheorem{example}[theorem]{Example}
\newtheorem{definition}[theorem]{Definition}
\newtheorem{exam}{Example}[section]

\hypersetup{
    colorlinks,
    citecolor=black,
    filecolor=black,
    linkcolor=black,
    urlcolor=black
}

\newcommand{\E}{\mathbb{E}}
\newcommand{\Var}{\mathrm{Var}}
\newcommand{\Cov}{\mathrm{Cov}}
\newcommand{\argmin}{\mathrm{argmin}}

\newcommand{\ybold}{\boldsymbol{y}}
\newcommand{\ystar}{\boldsymbol{y^{*}}}
\newcommand{\ytrbold}{\boldsymbol{y^{'}}}
\newcommand{\ytrstar}{\boldsymbol{y^{*'}}}
\newcommand{\ynew}{\boldsymbol{y^{new}}}

\newcommand{\fhatbold}{\boldsymbol{\hat{f}}}
\newcommand{\fhatstar}{\boldsymbol{\hat{f}^*}}
\newcommand{\betabold}{\boldsymbol{\beta}}

\newcommand{\betahatbold}{\boldsymbol{\hat{\beta}}}
\newcommand{\betahatstar}{\boldsymbol{\hat{\beta}^*}}
\newcommand{\mubold}{\boldsymbol{\mu}}
\newcommand{\mustar}{\boldsymbol{\mu^{*}}}
\newcommand{\mutrbold}{\boldsymbol{\mu^{'}}}
\newcommand{\mustartr}{\boldsymbol{\mu^{*'}}}

\newcommand{\abold}{\boldsymbol{a}}
\newcommand{\atrbold}{\boldsymbol{a^{'}}}
\newcommand{\bonebold}{\boldsymbol{b_{1}}}
\newcommand{\btwobold}{\boldsymbol{b_{2}}}
\newcommand{\btwostarbold}{\boldsymbol{b^{*}_{2}}}
\newcommand{\epsilonbold}{\boldsymbol{\epsilon}}
\newcommand{\epsilonstarbold}{\boldsymbol{\epsilon^{*}}}

\begin{document}

\def\spacingset#1{\renewcommand{\baselinestretch}%
{#1}\small\normalsize} \spacingset{1}


\if1\blind
{
  \title{\bf Assessing Prediction Error at Interpolation and Extrapolation Points}
  \author{Assaf Rabinowicz \thanks{
    The authors gratefully acknowledge \textit{Israeli Science Foundation, grant 1804/16}}\hspace{.2cm}\\
    Department of Statistics, Tel-Aviv University, Tel-Aviv, Israel, 69978\\
    \\
    Saharon Rosset\\
    Department of Statistics, Tel-Aviv University, Tel-Aviv, Israel, 69978}
  \maketitle
} \fi

\if0\blind
{
  \bigskip
  \bigskip
  \bigskip
  \begin{center}
    {\LARGE\bf Assessing Prediction Error at Interpolation and Extrapolation Points}
\end{center}
  \medskip
} \fi

\bigskip
\begin{abstract}
Common model selection criteria, such as $AIC$ and its variants, are based on in-sample prediction error estimators. However, in many applications involving predicting at interpolation and extrapolation points, in-sample error cannot be used for estimating the prediction error. In this paper new prediction error estimators, $tAI$ and $Loss(w_{t})$ are introduced. These estimators generalize previous error estimators, however are also applicable for assessing prediction error in cases involving interpolation and extrapolation. Based on the prediction error estimators, two model selection criteria with the same spirit as $AIC$ are suggested. The advantages of our suggested methods are demonstrated in simulation and real data analysis of studies involving interpolation and extrapolation in a linear mixed model framework.
\end{abstract}

\noindent%
\vfill

\newpage
\spacingset{1.45} 


\section{Introduction}
Predicting a phenomenon at different points than the points appearing in the training sample plays an important role across many research fields such as in Geostatistics \citep{li2014spatial,kyriakidis1999geostatistical}, health \citep{manton2012forecasting} and Econometrics \citep{baltagi2008forecasting}. In many of these use cases the new predicted points are interpolation or extrapolation points with respect to space or to time. For example, \cite{brown2002spatial} interpolated climate values in Southwestern U.S., where the coverage of climate information is sparse. By predicting at interpolation points, they created a high-resolution map of seasonal temperature and precipitation in this area. 
Another example given by \cite{stewart2009forecasting} is forecasting the  
effects of obesity and smoking on U.S. life expectancy in 2020 by using a data set for the years 2003 through 2006. 

Modeling approaches involving prediction at interpolation and extrapolation points were studied in Machine Learning, mainly in the context of transductive Support Vector Machine \citep{joachims1999transductive}, however also in regression \citep{le2006transductive}.

Assessing prediction error at interpolation and extrapolation points, or more generally at transduction points, cannot be done using traditional in-sample prediction error estimators as is used in $AIC$ \citep{akaike1974new} and its variants. Similarly, K-fold Cross-Validation, which estimates the generalization error, is also unsuitable in these cases, where prediction points are specified. 

This paper introduces a prediction error estimator, $tAI,$ which generalizes previous in-sample prediction error estimators like $mAI$ \citep{vaida2005conditional} and $cAI$ \citep{vaida2005conditional}, however, it doesn't assume that the predicted points are the same as the points appearing in the training sample and therefore is applicable to a wider range of use cases, such as cases involving prediction at interpolation and extrapolation points. Since prediction error assessment is highly related to model selection, a new model selection criterion, $tAIC,$ which is based on $tAI,$ is proposed as well. $tAI$ is suitable when the observations are normally distributed, whether they are correlated or not and therefore is applicable for various parametric models with different variance structure assumptions such as Linear Mixed Model (LMM), Gaussian Process Regression (GPR), Generalized Least Squares (GLS) and Linear Regression. Relaxing the normality requirement of $tAI$, we also propose in Section \ref{Optimism section} an approach for inference on interpolation and extrapolation that is based on squared error loss rather than likelihood, and hence generalizes the Optimism approach in model selection \citep{efron1986biased}.

In many use cases involving predicting at interpolation and extrapolation points, the dependent variable has a correlation structure \citep{li2014spatial,kyriakidis1999geostatistical}. For example, in the use case that is given by \cite{brown2002spatial}, it is natural to assume a spatial correlation structure on the Southwestern U.S. area. Similarly, in repeated measures studies that forecast long-term treatment effects, a correlation structure with respect to time is commonly assumed \citep{ho2011long}. Therefore, use cases involving correlated data and models that are implemented on correlated data, such as LMM, GPR and GLS, are good platforms for analyzing how predicting at interpolation and extrapolation points influences prediction error estimation and model selection. Before introducing $tAI,$ a setup which puts LMM, GPR and GLS under a unified framework, will be defined:

Let $\ybold\in\mathbb{R}^{n}$ and the fixed matrices $\{X\in\mathbb{R}^{n\times p},\,Z\in\mathbb{R}^{n\times q}\}$ be a training sample, $\ystar\in\mathbb{R}^{n^*}$ and the fixed matrices $\{X^*\in\mathbb{R}^{n^*\times p},\, Z^*\in\mathbb{R}^{n^*\times q}\}$ be a prediction set, where
\begin{align}
\ybold\sim N(\mubold,V),\,\ystar\sim N(\mustar,V^*),\label{Setup}
\end{align}
$\mubold=X\betabold,\,\mustar=X^*\betabold,$ $V$ is a 
function of $Z$ and $V^*$ is a function of $Z^*.$ For example, in LMM it is typically assumed that the columns of $Z,\,Z^*$ are associated with normally distributed random effects with covariance matrix $G\in\mathbb{R}^{q\times q}$ such that
\begin{align*}
V&=ZGZ^{'}+\sigma^2I_{n},\,\,\sigma\in \mathbb{R}^{+}\\    
V^{*}&=Z^{*}GZ^{*'}+\sigma^2I_{n^*},
\end{align*}
where $I_n,\,I_n^{*}$ are the identity matrices with dimensions $n$ and $n^*$ respectively. In GPR it is often assumed that 
\begin{align*}
V&=K(Z,Z)+\sigma^2I_n\\
V^{*}&=K(Z^{*},Z^{*})+\sigma^2I_n^{*}
\end{align*}
where $K$ is some kernel function.  

In addition denote 
\begin{align*}
R^*&=\Var(\ystar|\ybold)\\
R&=\Var(\ynew|\ybold),
\end{align*}
where $\ynew\sim N(\mubold,V)$ is an i.i.d copy of $\ybold.$

By normality of $\ybold$ and $\ystar,$ 
\[
\E(\ystar|\ybold)=X^{*}\betabold+\Cov(\ystar,\ybold)V^{-1}(\ybold-X\betabold).
\] 
Given $V,\,\Cov(\ystar,\ybold)$ and the ML estimator of $\betabold,$
\[
\betahatbold=(X^{'}V^{-1}X)^{-1}X^{'}V^{-1}\ybold,
\]
$\E(\ystar|\ybold)$ can be used for predicting $\ystar$ as follows \begin{align}
\fhatstar&=\hat{\E}(\ystar|\ybold) \label{Conditional Expectation} \\
&=X^{*}(X^{'}V^{-1}X)^{-1}X^{'}V^{-1}\ybold+\Cov(\ystar,\ybold)V^{-1}\left\{I_{n}-X(X^{'}V^{-1}X)^{-1}X^{'}V^{-1}\right\}\ybold. \nonumber
\end{align}

This procedure generalizes standard prediction practices in LMM, GPR and GLS. In addition, $\fhatstar$ is the Best Linear Unbiased Predictor (BLUP) \citep{harville1976extension}.

$tAI$ is an estimator of the following prediction error,
\begin{align}
-\frac{1}{n^{*}}\E_{\ystar|\ybold}l(\ystar)=-\frac{1}{n^{*}}\E_{\ystar|\ybold}\log\left[\frac{\exp{\left\{\frac{-1}{2}(\ystar-\fhatstar)^{'}R^{*^{-1}}(\ystar-\fhatstar)\right\}}}{\sqrt{(2\pi)^{n^*}|R^*|}}\right].\label{Transductive likelihood}
\end{align}
Correspondingly, given a set of candidate models, $tAIC$ would be defined as a model selection criterion selecting a model with the minimal $tAI.$ This methodology of estimating the prediction errors for different models and then selecting the model with the minimal prediction error, is the same as is implemented in $AIC$ and its variants.

$\{X^{*},Z^{*},R^*\}=\{X,Z,R\}$ is not assumed in the setup above and in its associated prediction error measure, eq. (\ref{Transductive likelihood}). Therefore, $tAI$ is applicable in various use cases that require flexibility in defining $\{X^{*},Z^{*},R^*\}.$ For example, in the use case mentioned above of \cite{brown2002spatial}, 
where GPR is used for predicting interpolated climate values (Kriging), it is reasonable to define $\{X^{*},Z^{*}\}$ as the data points at the high-resolution spatial array rather than as the data points at the training sample, $\{X,Z\},$ which cover the area sparsely. Therefore, while prediction error estimators that are based on in-sample error estimation and generalization error estimation are unsuitable to this case, $tAI$ is suitable. For similar considerations, $tAIC$ is required in repeated measures studies in health and Biomedicine, when the main interest is to select LMM model minimizing the prediction error at long-term points, $\{X^*,Z^*,R^{*}\}$, which are different than the points that are used for model building, $\{X,Z,R\}$ \citep{pope2002lung,li2008long}. 

Beside downscaling of climate maps and estimating long-term effect in clinical studies, interpolation and extrapolation using LMM and
Kriging are important tools for many research topics in mining engineering,
agriculture, environmental sciences, especially when sampling is difficult
and expensive like in mountainous and deep marine regions \citep{li2011review,stahl2006comparison,vicente2003comparative}. $tAI$ and $tAIC$ are relevant for
all these research topics as well as for others which don't involve interpolation and extrapolation but still don't satisfy $\{X^{*},Z^{*},R^{*}\}=\{X,Z,R\}.$
Various use cases will be presented and analyzed in Sections \ref{Use cases} and \ref{Numerical part}.

\section{$tAI$ and $tAIC$} \label{tAIC} 
$tAI$ is derived by estimating $-\E_{\ystar|\ybold}l(\ystar)/n^{*}$
by the averaged log-likelihood of the training sample,
\[
-\frac{1}{n}l(\ybold)=-\frac{1}{n}\log\left[\frac{\exp{\left\{\frac{-1}{2}(\ybold-\fhatbold)^{'}R^{-1}(\ybold-\fhatbold)\right\}}}{\sqrt{(2\pi)^{n}|R|}}\right],
\]
plus a penalty correction
\[
C_{tAI}=\E_{\ybold}\left[-\frac{1}{n^{*}}\E_{\ystar|\ybold}l(\ystar)-\left\{-\frac{1}{n}l(\ybold)\right\}\right],
\]
where 
\[
\fhatbold=X\betahatbold+(V-R)V^{-1}(\ybold-X\betahatbold)
\]
is the estimated conditional expectation, $\hat{\E}(\ystar|\ybold),$ when $\{X^*,Z^*,R^*\}=\{X,Z,R\}.$
This approach of estimating prediction error by deriving the bias of the training error is also used in $AIC$ and its variants \citep{akaike1974new}. Consequently, the estimator
\[
tAI=-\frac{1}{n}l(\ybold)+C_{tAI}
\]
doesn't contain $\ystar$ but still satisfies
\[
E_{\ybold}tAI=\E_{\ybold}\left\{-\frac{1}{n^{*}}\E_{\ystar|\ybold}l(\ystar)\right\},
\]
and $tAI$ can be seen either as an estimator of 
$-\E_{\ystar|\ybold}l(\ystar)/n^{*}$ or of its expectation  $-\E_{\ybold}\E_{\ystar|\ybold}l(\ystar)/n^{*}=-\E_{\ybold,\ystar}l(\ystar)/n^{*}.$ 

The following theorem and corollary introduce a general expression for $C_{tAI}$ and therefore also for $tAI$ and $tAIC$.

\begin{theorem} \label{Main Theorem}
Consider the setup given in eq. (\ref{Setup}). In addition, let $H\ybold\in\mathbb{R}^{n}$ and $H^{*}\ybold\in\mathbb{R}^{n^{*}}$ be predictors of $\ybold$ and $\ystar$ respectively, where $H$ and $H^*$ don't contain $\ybold,\,\ystar$ and satisfy $H\mubold=\mubold,\,H^{*}\mubold=\mustar.$ Then 
\begin{align*}
\E_{\ybold}\left\{-\frac{1}{n^{*}}\E_{\ystar|\ybold}l(\ystar)+\frac{1}{n}l(\ybold)\right\}=&\frac{1}{n}tr\left(R^{-1}HV\right)-\frac{1}{n^*}tr\left(R^{*^{-1}}H^{*}\Cov(\ybold,\ystar)\right)\\
 &+\frac{1}{2}\left\{\log\left(\frac{|R^{*}|^{\frac{1}{n^*}}}{|R|^{\frac{1}{n}}}\right)+\frac{1}{n^*}tr\left(R^{*^{-1}}V^{*}\right)-\frac{1}{n}tr\left(R^{-1}V\right)\right\}\\
 &+\frac{1}{2}\left\{\frac{1}{n^*}tr\left(R^{*^{-1}}H^{*}VH^{*'}\right)-\frac{1}{n}tr\left(R^{-1}HVH^{'}\right)\right\}.
\end{align*}
\end{theorem}
The proof is attached in Appendix \ref{$tAIC$ derivation appendix}.

\begin{corollary} \label{C_{tAI}}
Under the set-up described in eq. (\ref{Setup}), the conditions in Theorem \ref{Main Theorem} are satisfied by the BLUPs $\fhatbold$ and $\fhatstar,$ where
\begin{align*}
\fhatbold & =H\ybold\\
H & =X(X^{'}V^{-1}X)^{-1}X^{'}V^{-1}+(V-R)V^{-1}\left\{I_{n}-X(X^{'}V^{-1}X)^{-1}X^{'}V^{-1}\right\},
\end{align*}
and
\begin{align*}
\fhatstar & =H^{*}\ybold\\
H^{*} & =X^{*}(X^{'}V^{-1}X)^{-1}X^{'}V^{-1}+\Cov(\ystar,\ybold)V^{-1}\left\{I_{n}-X(X^{'}V^{-1}X)^{-1}X^{'}V^{-1}\right\}.
\end{align*}
\end{corollary}

By Corollary \ref{C_{tAI}} and Theorem \ref{Main Theorem}, $C_{tAI}$ can be calculated under the setup that is described in eq. (\ref{Setup}). Therefore, $tAI$ can be implemented in LMM, GPR and other related models.

Beside prediction error estimation, these results can be used for defining the following model selection criterion.
\begin{definition} Given set of models $\mathcal{H},$ satisfying the conditions in Theorem \ref{Main Theorem}, $tAIC$ is the following criterion
\[
h_{best}=\underset{h\in\mathcal{H}}{\argmin}\, tAI_h,
\]
where $tAI_h$ is $tAI$ for model $h.$
\end{definition}

\subsection{Comparison with other prediction error estimators} \label{Comparison with mAIC and cAIC}
The prediction error estimators that appear in $cAIC$ and
$mAIC$ \citep{vaida2005conditional} were developed for normal linear models under different restrictions on the variance structure, but assuming $\{X,Z,R\} = \{X^*,Z^*,R^*\}$. $cAIC$ is aimed at the LMM and GPR case, where $\Cov(\ystar,\ybold) \neq 0$, while $mAIC$ considers the GLS case where  $\Cov(\ystar,\ybold) = 0$. For $cAIC$ the prediction error estimate is: 
\[
cAI=-\frac{1}{n}l(\ybold)+\frac{tr(H)}{n},\,\, \text{\footnotemark}
\]
\footnotetext{\cite{vaida2005conditional} define this prediction error estimator with a factor of $2n,$ i.e., as $2n\times cAI.$ In addition, they denote the prediction error estimator as cAIC. However, here, in order to distinguish between the prediction error estimator and the model selection procedure, the prediction error estimator is denoted as $cAI$ and the criterion as $cAIC.$ Similarly with $mAI$ and $mAIC$.} 
while for $mAIC$ it is: 
\[
mAI=-\frac{1}{n}l(\ybold)+\frac{p}{n}.
\]
$cAIC,\;mAIC$ are defined from $cAI,\;mAI$ similarly to $tAIC.$ 

It is easy to confirm that when $\{X,Z,R\} = \{X^*,Z^*,R^*\}$ our $tAI$ formula indeed reduce to the $cAI$ and $mAI$ formulas.

In addition, for GLS, we can also show an interesting interpretation for the difference between the $mAI$ and $tAI$ expressions. With a little algebra we get:

\begin{align*}
C_{tAI}(GLS)=&\frac{p}{n}+\frac{1}{2}\log\left(\frac{|V^{*}|^{\frac{1}{n^{*}}}}{|V|^{\frac{1}{n}}}\right)\\
&
+\frac{1}{2}tr\left\{(X^{'}V^{-1}X)^{-1}\left(\frac{1}{n^{*}}X^{*'}V^{*^{-1}}X^{*}-\frac{1}{n}X^{'}V^{-1}X\right)\right\}\\
=&\frac{p}{n}+\frac{1}{2}\log\left(\frac{|V^{*}|^{\frac{1}{n^{*}}}}{|V|^{\frac{1}{n}}}\right)
+\frac{1}{2}tr\left[\Var(\betahatbold)\left\{\frac{1}{n^{*}}\Var(\betahatstar)^{-1}-\frac{1}{n}\Var(\betahatbold)^{-1}\right\}\right],
\end{align*}
where 
\[\betahatstar=(X^{'*}V^{*^{-1}}X^{*})^{-1}X^{'*}V^{*^{-1}}\ystar.
\]
Since $\Var(\betahatbold)$ achieves the Cramer-Rao bound:
\begin{align}\label{c_mAI and c_tAI}
C_{tAI}(GLS) & =\frac{p}{n}+\frac{1}{2}\log \left(\frac{|V^{*}|^{\frac{1}{n^{*}}}}{|V|^{\frac{1}{n}}}\right)+\frac{1}{2}tr\left[\mathcal{I}(\betahatbold)^{-1}\left\{\frac{1}{n^{*}}\mathcal{I}(\betahatstar)-\frac{1}{n}\mathcal{I}(\betahatbold)\right\}\right],
\end{align}
where $\mathcal{I}$ is Fisher-information. The determinants $|V|$ and $|V^{*}|$ are often called the generalized variance \citep{wilks1932certain,johnson2014applied}.

\subsection{Relaxing Theorem \ref{Main Theorem} Conditions}
Although this paper focuses on prediction error estimation and model selection for LMM and GPR, Theorem \ref{Main Theorem} is more general and doesn't assume the paradigm applied in LMM and GPR, i.e., predicting using $E(\ystar|\ybold)$ and estimating the marginal mean parameters with MLE. Theorem \ref{Main Theorem} assumes:
\begin{enumerate}
\item Normality of  $\ystar$ and $\ybold.$ 
\item $\E \ybold=X\betabold,\,\E \ystar=X^*\betabold$ 
\item $H\mubold =\mubold,\,H^*\mubold =\mustar$ \label{Bias conditions} 
\end{enumerate}
and therefore can be used in other cases satisfying the above conditions.

When the normality assumption cannot be taken, another model selection criterion, which is based on similar approach as $tAI$ can be implemented. For more details see Section \ref{Optimism section}.

In case the normality assumption can be taken, however the fitted model doesn't satisfy condition \ref{Bias conditions} of unbiasedness, the following extended version of Theorem \ref{Main Theorem} results can be used instead:
\begin{align*}
&\E_{\ybold}\left\{-\frac{1}{n^{*}}\E_{\ystar|\ybold}l(\ystar)+\frac{1}{n}l(\ybold)\right\}\\
&\qquad\qquad=\frac{1}{n}tr\left(R^{-1}HV\right)-\frac{1}{n^*}tr\left(R^{*^{-1}}H^{*}\Cov(\ybold,\ystar)\right)\\
 & \qquad\qquad\,\,\,\,\,\,\,+\frac{1}{2}\left\{\log\left(\frac{|R^{*}|^{\frac{1}{n^*}}}{|R|^{\frac{1}{n}}}\right)+\frac{1}{n^*}tr\left(R^{*^{-1}}V^{*}\right)-\frac{1}{n}tr\left(R^{-1}V\right)\right\} \\
 &\qquad\qquad\,\,\,\,\,\,\,+\frac{1}{2}\left\{\frac{1}{n^*}tr\left(R^{*^{-1}}H^{*}VH^{*'}\right)-\frac{1}{n}tr\left(R^{-1}HVH^{'}\right)\right\}\\
  & \qquad\qquad\,\,\,\,\,\,\,+\frac{1}{2n}tr\left(R^{-1}(2H\mubold\mutrbold-\mubold\mutrbold-H\mubold\mutrbold H^{'})\right)\\
& \qquad\qquad\,\,\,\,\,\,\,\,  -\frac{1}{2n^*}tr\left(R^{*^{-1}}(2H^{*}\mubold\mustartr-\mustar\mustartr-H^{*}\mubold\mutrbold H^{*'})\right).
\end{align*}
The proof can be found in Appendix \ref{$tAIC$ derivation appendix} as part of the proof of Theorem \ref{Main Theorem}.

Note that this expression is less useful, as it depends on $\mubold$ and $\mustar$ which are unknown.

\section{Use cases} \label{Use cases}
In this section, typical use cases of using $tAI$ and $tAIC$ are presented.
\subsection{Predicting interpolation and extrapolation in spatial array and longitudinal temporal data}
As was described in the introduction, predicting interpolated and extrapolated data points using LMM and GPR is common in Biomedicine, health, Climatology and other research fields, where temporal and spatial datasets are common. The flexible definition of $X^*,Z^*,R^*$ and $V^*$ in $tAI$ makes it applicable when the goal is to estimate prediction error at interpolated and extrapolated data points along time and space dimensions.

In Section \ref{Numerical part} we analyze numerically a repeated measures clinical study, containing child growth measurements \citep{potthoff1964generalized}, where interpolation and extrapolation objectives can be defined and application of $tAI$ is demonstrated. The following example, built on the application of \cite{tsanas2010accurate}, demonstrates that appropriate use of $tAIC$ can also simplify and improve on existing methodology.
\begin{exam}\label{Interpolation example}
Tsanas et al. introduced a new method for measuring progression of Parkinson's disease. Their motivation is that the standard methodology for measuring Parkinson progression, which uses UPDRS score (Unified Parkinson's Disease Rating Scale), is costly and requires a physician visit. Their alternative methodology is creating a formula that approximates the UPDRS score with speech signals which are not costly. Six months data was collected for their study, containing large amount of longitudinal speech signal measurements per patient, however, UPDRS scores were collected only at a small number of the time points. In order to select the best covariates with respect to the whole speech signals sample, they suggested to interpolate the UPDRS scores using 'straightforward linear interpolation', then to fit several alternative models and to select one of them using $AIC$ and other model selection criteria. 
An alternative paradigm that doesn't require imputing UPDRS score is by using $tAIC.$ Since $tAIC$ doesn't assume $\{X^{*},Z^{*}\}=\{X,Z\},$ there is no need in interpolating the UPDRS score in order to select a model minimizing the estimated prediction error with respect to the whole speech signals sample.
\end{exam}

We note that in Example \ref{Interpolation example}, one may think that $\ystar$ is used twice, for model building and for prediction error estimation, and therefore over-fitting can occur.
However, since in $tAI$ approach, unlike in cross-validation approach,
$\ystar$ is used as a conceptual idea in order to derive $C_{tAI}$ and
not as real observations, $\ystar$ is not used twice. 

In the spatial data analysis domain, common application areas include geographical data \citep{li2014spatial} and neuroimaging data \citep{salimi2011using}. Such studies usually use GPR rather than LMM. Although GPR and LMM reflect different perspectives --- while GPR is
based on functional data analysis, LMM is based on multivariate analysis --- and use different techniques for expressing the covariance matrices, both models use conditional expectation, $\fhatstar,$ for prediction, hence $tAI$ is also applicable for GPR. In the Introduction we demonstrated this by the use case of creating high-resolution climate maps \citep{brown2002spatial}. A similar use case, containing chemical concentration in soil data is analyzed numerically in Section \ref{Numerical part}.

\subsection{Other Transductive Settings}
LMM and GPR are also used for modeling data without spatial or temporal correlation structure, and the prediction problems that arise often involve prediction outside the training sample. 

One interesting example is modeling the effect of SNPs (Single Nucleotide Polymorphism) on a phenotype as part of a Genome-Wide Association Study (GWAS). In this case the common practice is to consider the SNPs as random effects and other explanatory variables (e.g. age, height and gender) as fixed effects \citep{zhang2010mixed}.  When using LMM for modeling the effect of SNPs on phenotype, $tAI$ allows estimating the prediction error for an extended population compared to the training sample. It is directly useful in the important case when $\{X^*,Z^*\}$ can be collected from other studies which investigate different phenotype, however contain the SNPs and the explanatory variables that are used in the training sample \citep{wray2013pitfalls}. 

Missing values of the dependent variable which is a common phenomenon in statistical analysis and in particular in clinical trial with repeated measures study design \citep{wood2004missing,o2012prevention}.
There are many methods for handling missing values in repeated measures studies, some of the methods involving missing values imputation \citep{mallinckrodt2003assessing}. In case of having missing data of the dependent variable at some known points but the goal is to estimate the prediction error with respect to the original study design \citep{hogan2004handling}, $tAI$ can be used without imputing the missing values.

\section{Numerical Results} \label{Numerical part}
This section focuses on comparison between $tAI,\,cAI$ and $mAI,$ as well as between their corresponding model selection criteria, $tAIC,\,cAIC$ and $mAIC,$ using simulation and real data analyses.

\subsection{Simulation Analyses} \label{Simulation analysis}
The goal of the following analyses is to investigate the accuracy of $tAI,cAI$ and $mAI$ in estimating $-\E_{\ystar|\ybold}l(\ystar)/n^{*},$ for different sample sizes and variance setups.
In addition, $tAIC, cAIC$ and $mAIC$ will also be analyzed and compared with respect to the oracle solution
\begin{align*}
h_{best}&=\underset{h\in\mathcal{H}}{\argmin}-\frac{1}{n^{*}}\E_{\ystar|\ybold}l_h(\ystar).
\end{align*}
Additional numerical results with respect to a potentially different oracle solution 
\begin{align*}
h_{best}&=\underset{h\in\mathcal{H}}{\argmin}-\frac{1}{n^{*}}\E_{\ybold}\E_{\ystar|\ybold}l_h(\ystar)
\end{align*}
are presented in Appendix \ref{Additional Numeriocal Results}.
\paragraph{Simulation setup}
The simulation demonstrates prediction error estimation and model selection for the following LMM setting:
\begin{align*}
\phi_{i,j}=0.5\times time_{i,j}+\sum_{k=0}^{k=2}x_{i,j,k}+2\times\sum_{k=3}^{k=6}x_{i,j,k}+b_{i,1}+time_{i,j}\times b_{i,2}+\epsilon_{i,j},
\end{align*}
where $i\in\{1,...,S\}$ is the subject number and $j\in\{1,...,12\}$ is the measurement number. $b_{i,1}$ is distributed $N(0,15),$ $b_{i,2}$ is distributed $N(0,1)$ and $\epsilon_{i,j}$ is distributed $N(0,\sigma^2).$ In addition, $time_{i,j}=j,\,\forall j\leq 10,$ $time_{i,11}=15,\,time_{i,12}=20,$
$x_{i,j,0}=1,$ $x_{i,j,1}$ was drawn from $Ber(0.5)$ and $x_{i,j,k},\,\forall k\geq2$ was drawn form $N(0,1).$

The dependent variable in the training set $,\ybold,$ was defined as $\phi_{i,j}\,\forall j\in\{1,..,10\},$
the dependent variable in the prediction set, $\ystar,$ was defined as $\phi_{i,j}\,\forall j\in\{11,12\}.$
Therefore, this setting demonstrates predicting at extrapolation time points. This setting was generated nine times, for different number of subjects, $S\in\{10,20,30\},$ and different residual variance values, $\sigma^2\in\{15,20,25\}.$ Each simulation was repeated $200$ times.

Three linear mixed models were fitted given the true covariance matrices, all the models contain the time covariate, in addition, model number 1 contains $x_{i,j,k},\,\forall k\leq2,$ model number 2 contains $x_{i,j,k},\,\forall k\leq4$ and model number 3 contains $x_{i,j,k},\,\forall k$ which is also the model that generates the data.
\paragraph{Results}
Figure \ref{Densities} presents the densities of  $tAI,\,cAI,\,mAI$ and $-\E_{\ystar|\ybold}l(\ystar)/n^{*}$ as a function of the sample size and $\sigma^2$ for model number $3,$ as generated from the $200$ simulation runs.
\begin{figure}[h!]
\begin{centering}
\includegraphics[width=1\linewidth]{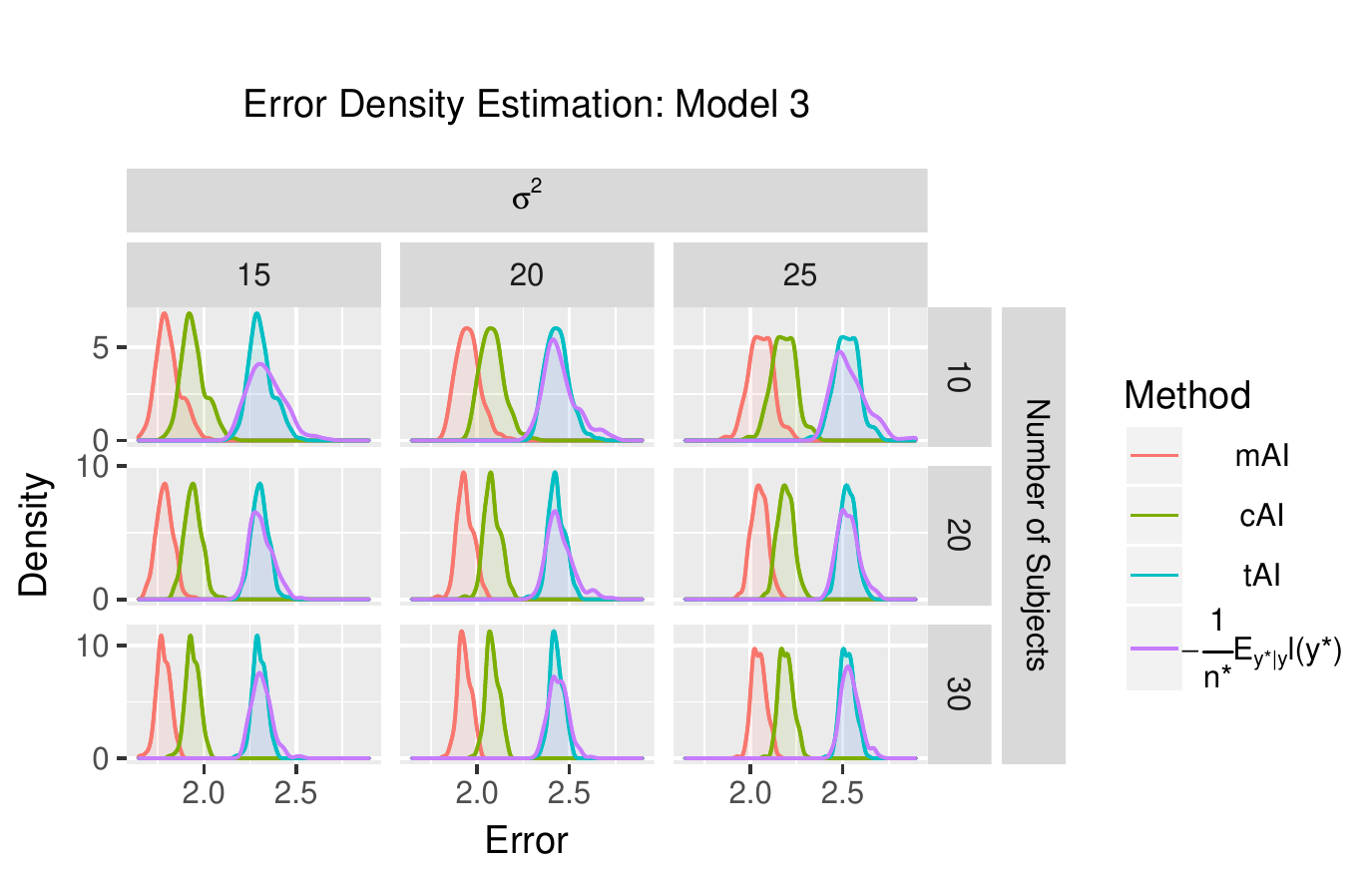}
\caption{Densities of  $tAI,\,cAI,\,mAI$ and $-\E_{\ystar|\ybold}l(\ystar)/n^{*}$ as a function of the number of subjects and $\sigma^2$.}
\label{Densities}
\end{centering}
\end{figure}

As can be seen from Figure \ref{Densities}, $tAI$ density is concentrated around the mean of $-\E_{\ystar|\ybold}l(\ystar)/n^{*}.$ $cAI$ and $mAI$ are stochastically smaller than $-\E_{\ystar|\ybold}l(\ystar)/n^{*}$ since their corrections, $tr(H)/n$ and $p/n,$ are unsuitable for this case of predicting at extrapolation points.

In addition, since $tAI,\,cAI$ and $mAI$ share the same random part, $l(\ybold)/n,$ but different mean, $-\E_{\ybold}l(\ybold)/n$ plus $C_{tAI},\,tr(H)/n,\,p/n$ respectively, their densities have the same shape however shifted with respect to the corrections. In contrast, $-\E_{\ystar|\ybold}l(\ystar)/n^{*}$ has the same mean as $tAI$ but different variance, since $\Var\left(-\E_{\ystar|\ybold}l(\ystar)/n^{*}\right)$ depends on $H^*,\,R^*$ and $n^*$ that do not appear in $\Var\left(tAI\right)=\Var\left(-l(\ybold)/n\right).$ In our case, $H^*$ contains large values compared to $H$ and therefore
\[
\Var\left(\E_{\ystar|\ybold}-\frac{1}{n^{*}}l(\ystar)\right)>\Var\left(tAI\right).
\]

Figure \ref{Average Argmin} presents for each criterion, $tAIC,\,cAIC$ and $mAIC,$ the error 
\[
\E_{\ystar|\ybold}-\frac{1}{n^{*}}l_{h_{best}}(\ystar),
\]
where $h_{best}$ is the selected model by the relevant criterion. This error reflects the true average error that is obtained when implementing the different model selection criteria. In addition, the average error of the oracle criterion,
\begin{align*}
h_{best}&=\underset{h\in\{1,2,3\}}{\argmin}-\frac{1}{n^{*}}\E_{\ybold}\E_{\ystar|\ybold}l_h(\ystar),
\end{align*}
is presented as well.
\begin{figure}[h!]
\begin{centering}
\includegraphics[width=1\linewidth]{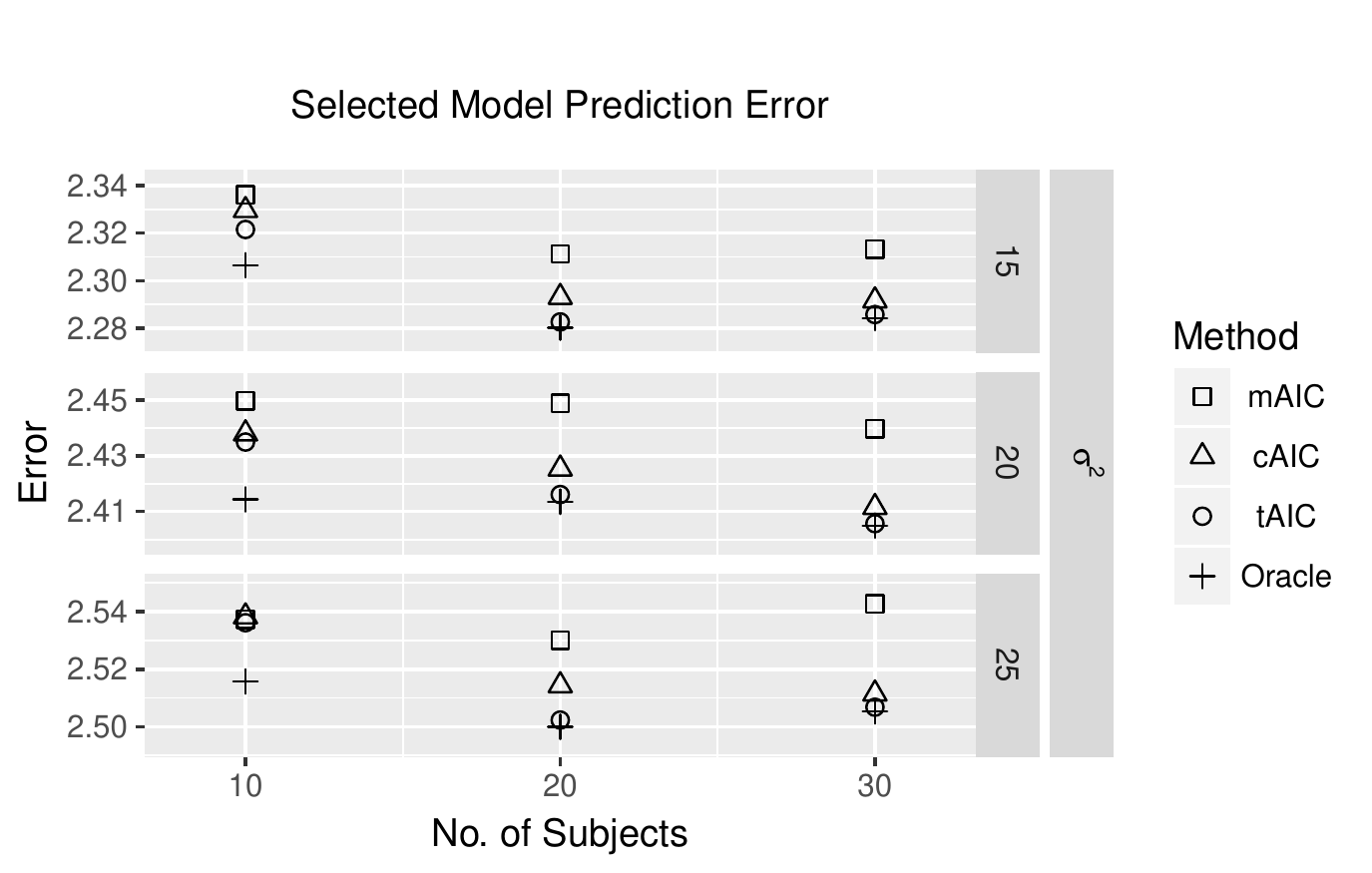}
\caption{For each setup, each symbol refers to the prediction error $\E_{\ystar|\ybold}-l_{h_{best}}(\ystar)/n^{*}$ of the relevant criterion, $mAIC,\,cAIC\,tAIC$ and the oracle criterion.}
\label{Average Argmin}
\end{centering}
\end{figure}

As can be seen from Figure \ref{Average Argmin}, $tAIC$ obtain better results than $cAIC$ and $mAIC$ in all the nine setups. Similar analysis with respect to the error
\[
\E_{\ybold}\E_{\ystar|\ybold}-\frac{1}{n^{*}}l_{h_{best}}(\ystar
)\]
is presented in in Appendix \ref{Additional Numeriocal Results}.

Figure \ref{Agreement Rate} presents the agreement rate of the criteria, $tAIC,\,cAIC$ and $mAIC$ with the oracle criterion 
\begin{align*}
h_{best}&=\underset{h\in\{1,2,3\}}{\argmin}-\frac{1}{n^{*}}\E_{\ystar|\ybold}l_h(\ystar).
\end{align*}
\begin{figure}[h!]
\begin{centering}
\includegraphics[width=1\linewidth]{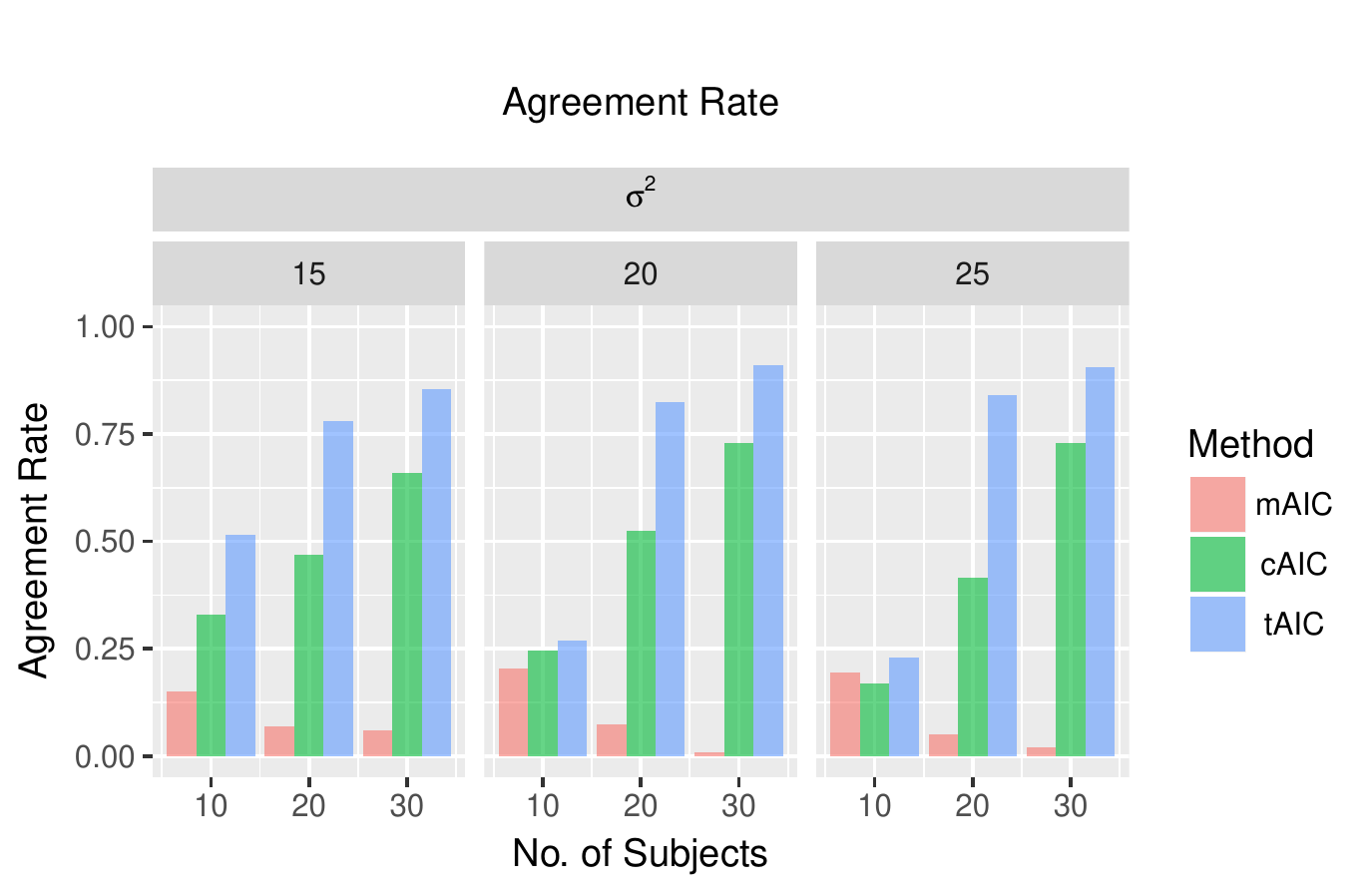}
\caption{For each setup, each bar refers to the agreement rate of the relevant criterion with the oracle criterion}
\label{Agreement Rate}
\end{centering}
\end{figure}

As can be seen from Figure \ref{Agreement Rate}, $tAIC$ achieves the best results in this case as well.

Similar analysis with respect to the oracle criterion
\begin{align*}
h_{best}&=\underset{h\in\{1,2,3\}}{\argmin}-\frac{1}{n^{*}}\E_{\ybold}\E_{\ystar|\ybold}l_h(\ystar)
\end{align*}
is presented in Appendix \ref{Additional Numeriocal Results}.

\subsection{Real data analyses}
The analyses below focus on comparison between $tAI,\,cAI,\,mAI$ and
\[
-\frac{1}{n^{*}}l(\ystar).
\]
Here, $-l(\ystar)/n^{*}$ is used as a ground truth instead of $-\E_{\ystar|\ybold}l(\ystar)/n^{*}$ since the latter is unknown for the real data sets.
\subsubsection{Meuse data}
\paragraph{Data description}
Meuse data set was introduced by \cite{rikken1993soil} and is available in R software. The data was collected in a floodplain area of the river Meuse, near the village of Stein, Netherlands, and contains 155 measurements of topsoil concentrations of Zinc, Lead, Copper and Cadmium, along with location (latitude and longitude) and other covariates. In addition, another data set, Meuse.grid, is analyzed. Meuse.grid is a higher resolution grid of the same area, containing 3103 observations of location and some of the covariates that are available in the Meuse data set, however it doesn't contain the metal concentration measurements. The Meuse.grid is available in R software as well. 
\paragraph{Results}
The Meuse data set was
partitioned randomly into training and test samples. Four Gaussian process regression models were fitted to the log of the Lead concentration.\footnote{Only $\log(Lead)$ can be analyzed under the normality assumption.} All the models share the same kernel structure, squared-exponential kernel,
\[
K(\boldsymbol{Z_i},\boldsymbol{Z_j})=\sigma_{f}^{2}\exp\left[-\frac{1}{2}\left\{\frac{1}{l_{1}^{2}}\left(Z_{i,1}-Z_{j,1}\right)^{2}+\frac{1}{l_{2}^{2}}\left(Z_{i,2}-Z_{j,2}\right)^{2}\right\}\right],
\]
where $Z_{i,1}$ refers to the latitude of measurement $i,$ $Z_{i,2}$ refers to longitude of measurement $i$ and $l_1,\,l_2$ and $\sigma_{f}$ lie in $\mathbb{R}^{+}.$ Each model has a different marginal mean, see Table \ref{Table Meuse}. The descriptions of the covariates can be found in R software.
\begin{table}[h!]
  \centering
  \caption{Meuse data: Covariates}  \label{Table Meuse}
    \begin{tabular}{|c| c c c|} \hline
   Model & \multicolumn{3}{c|}{Covariates} \\ \hline
       & $Intercept,dist,ffreq, soil$     & $dist \times ffreq$     &$dist \times Soil$\\ \hline
    1  & \checkmark  &              &             \\
    2  & \checkmark  & \checkmark   &              \\
    3  & \checkmark  &              & \checkmark   \\
    4  & \checkmark  & \checkmark   & \checkmark  \\\hline
    \end{tabular}
\end{table}
\begin{figure}[h!]
\begin{centering}
\includegraphics[width=1\linewidth]{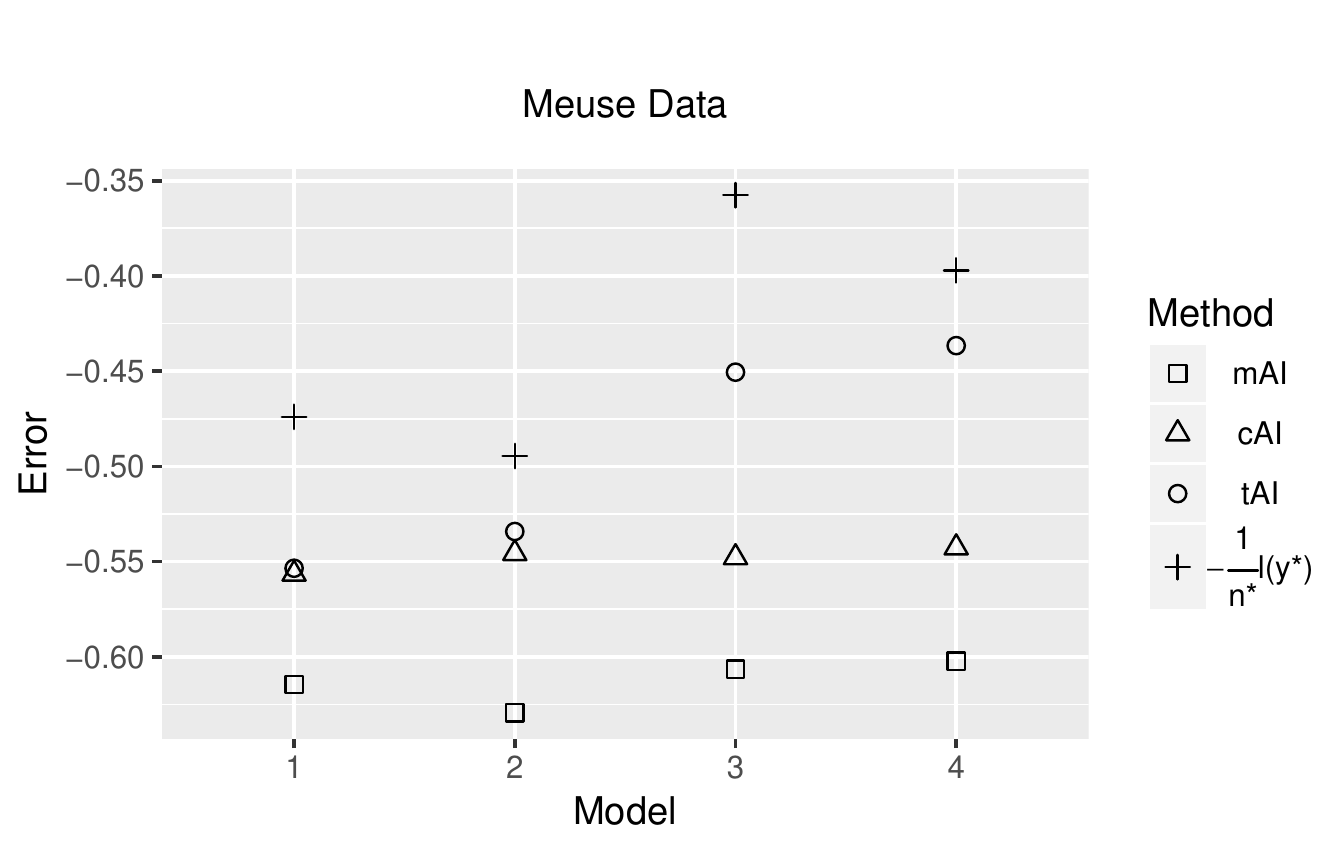}
\caption{For each model, each symbol refers to a prediction error, estimated by a different method.}
\label{Meuse}
\end{centering}
\end{figure}
As can be seen in Figure \ref{Meuse}, $tAI$ estimates $-l(\ystar)/n^{*}$ most accurately. The other prediction error estimators consistently under estimate $-l(\ystar)/n^{*}.$ 

Figure \ref{Meuse.grid} is based on Meuse and on Meuse.grid data sets where the whole Meuse data set is used as training data and the Meuse.grid data set is used as the prediction set, $\{X^*,Z^*\}.$ Since the Lead consternation is not given in the Meuse.grid data set, then $-l(\ystar)/n^{*}$
is unknown. Therefore $tAI,\,cAI$ and $mAI$ are compared without having a ground truth. 
\begin{figure}[h!]
\begin{centering}
\includegraphics[width=1\linewidth]{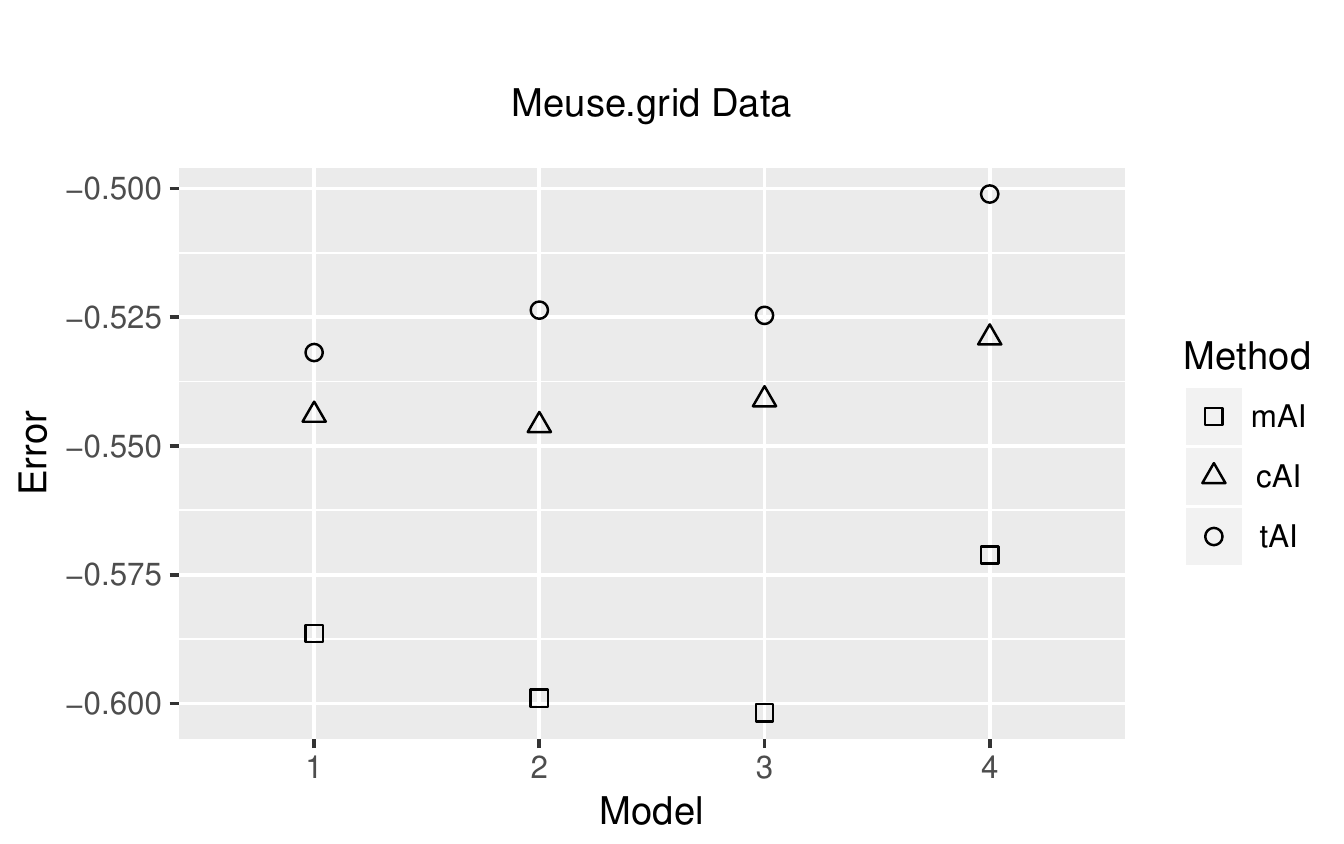}
\caption{For each model, each symbol refers to a prediction error, estimated by a different method.}
\label{Meuse.grid}
\end{centering}
\end{figure}

As can be seen from Figure \ref{Meuse.grid}, the differences between the $tAI,\,cAI$ and $mAI$ are sustained and the results are consistent with the previous figures, i.e., $cAI,\,mAI$ give lower error estimates, which likely underestimate the prediction error.

\subsubsection{Growth data}
\paragraph{Data description}
The Growth data was introduced by \cite{potthoff1964generalized} and contains four skull length measurements for 27 children at ages $8,\,10,\,12$ and $14$ (total of $27\times4$ measurements) along with the child's age and gender. 
\paragraph{Results}
Figure \ref{Growth14} presents a scenario where the training sample is defined as the skull length measurements at ages $8,\,10,\,12$ and the prediction set is defined as the skull length measurements at age $14.$ Three linear mixed models are fitted, all have the same variance structure, containing random intercept per child and random slope for the child's age, however each model has a different set of fixed effects (see Table \ref{Growth table}).
\begin{table}[h!]
  \centering
  \caption{Growth data: Covariates}  \label{Growth table}
    \begin{tabular}{|c| c c c c|} \hline
   Model & \multicolumn{4}{c|}{Covariates} \\ \hline
     & $Intercept$     & $Age$     &$Gender$   &$Age\times Gender$ \\ \hline
    1  & \checkmark  & \checkmark   &            &\\
    2  & \checkmark  & \checkmark   & \checkmark & \\
    3  & \checkmark  & \checkmark   & \checkmark & \checkmark \\\hline
    \end{tabular}
\end{table}
\begin{figure}[h!] 
\begin{centering}
\includegraphics[width=1\linewidth]{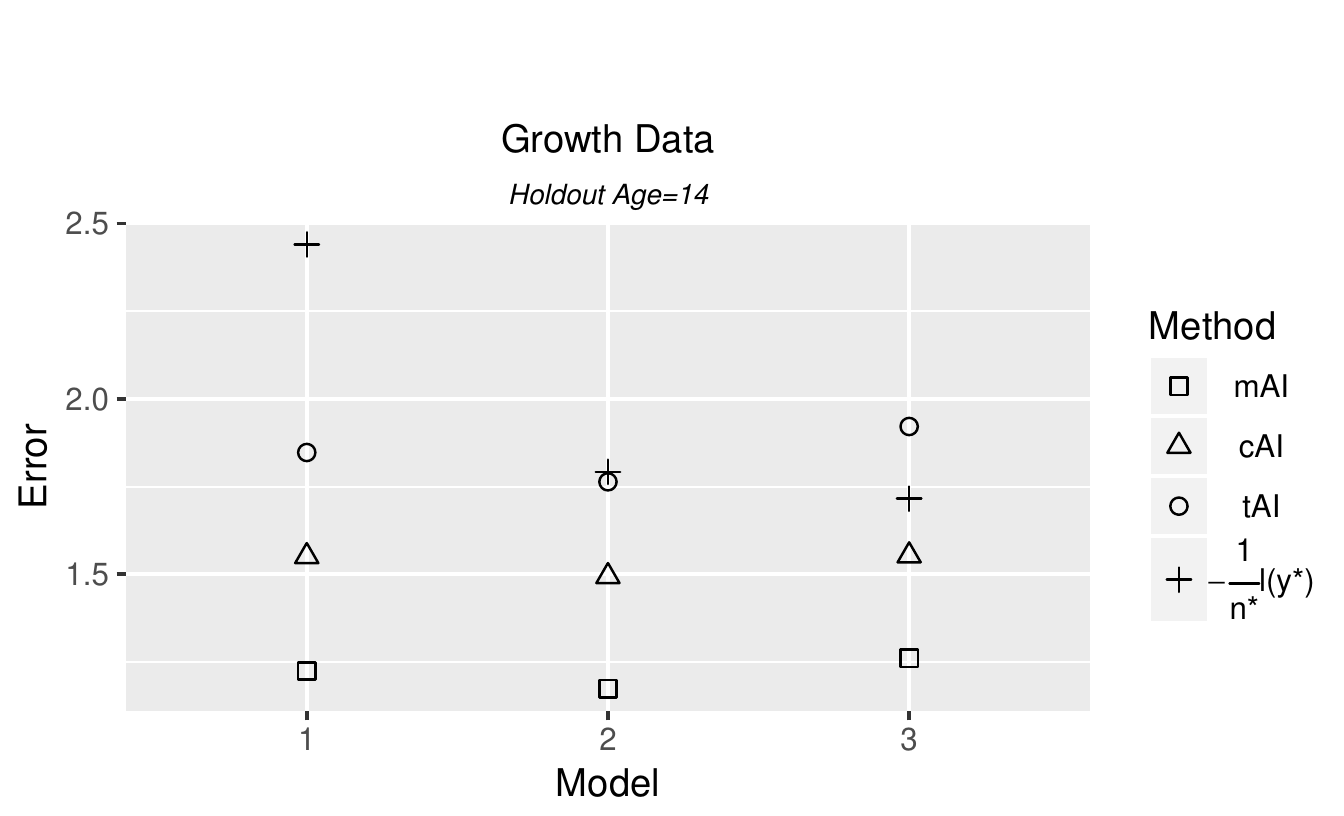}
\caption{For each model, each symbol refers to a prediction error, estimated by a different method.}\label{Growth14}
\end{centering}
\end{figure}

As can be seen in Figure \ref{Growth14}, in general perspective, $tAI$ estimates $-l(\ystar)/n^{*}$ most accurately. The other prediction error estimators under estimate $-l(\ystar)/n^{*}.$ 

Figure \ref{GrowthNot14} presents three similar analyses as is presented in Figure \ref{Growth14}, however where the other time-points measurements are designated as holdout.
\begin{figure}[h!] 
\begin{centering}
\includegraphics[width=1\linewidth]{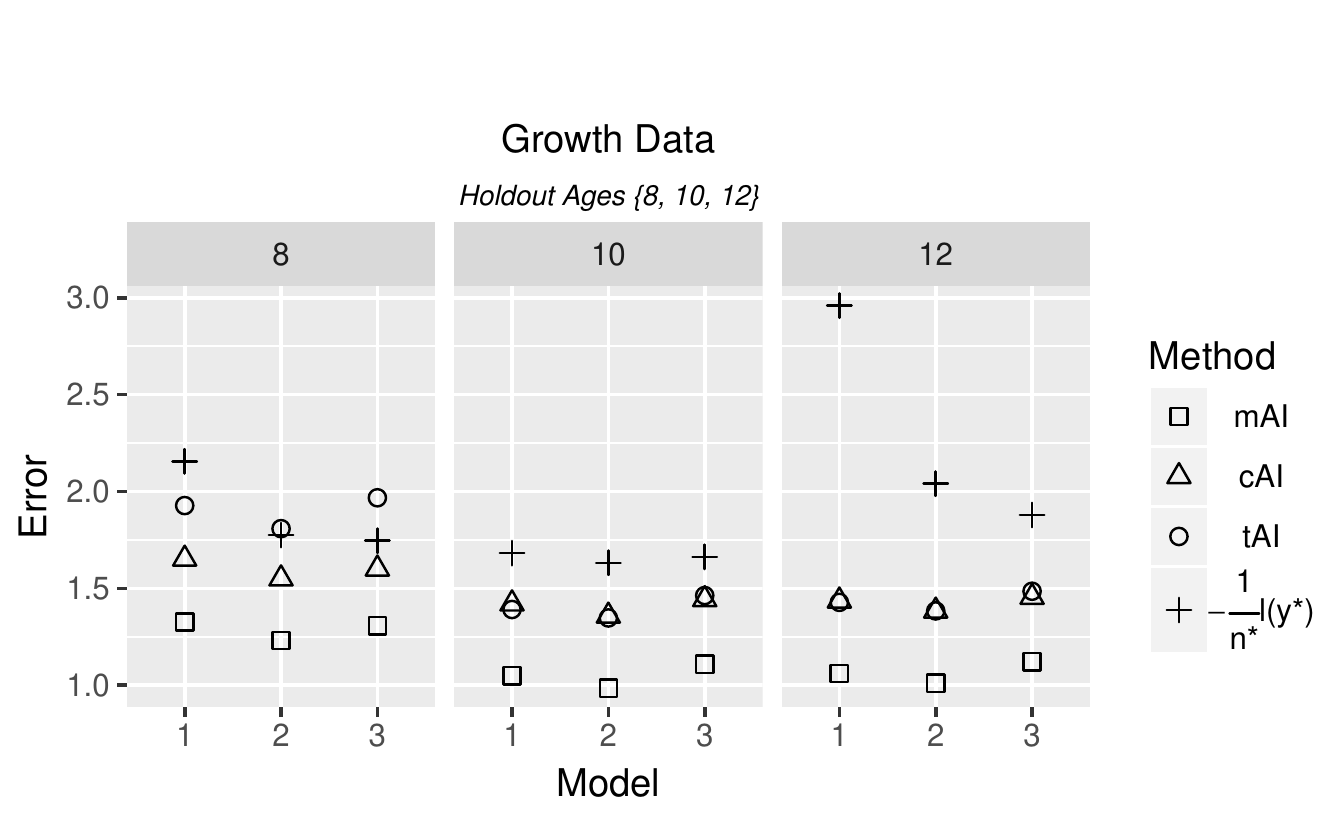}
\caption{For each model, each symbol refers to a prediction error, estimated by a different method.} 
\label{GrowthNot14}
\end{centering}
\end{figure}
When $age=8,$ the results are similar to the results in Figure \ref{Growth14}, however, when $age=10$ and $age=12,$ $tAI$ and $cAI$ have similar performance. This is not surprising since in these cases $\{X^*,Z^*,R^*\}$ is similar to $\{X,Z,R\}.$


\section{Optimism for Prediction at Interpolation and Extrapolation Points}
\label{Optimism section}
The formulation of $tAI$ and the derivation of $C_{tAI}$ are based on the normality assumptions of $\ybold$ and $\ystar$ which is commonly assumed when LMM and GPR are implemented. However, the approach that is used for developing $tAI$ can be used for creating other prediction error estimators that are not based on the normality assumption of $\ystar$ and $\ybold.$ For example, in the standard formulation of the prediction error estimator that is based on expected Optimism correction \citep{efron1986biased},
\[
Loss(Opt)=\frac{1}{n}\|\ybold-H\ybold\|^2_2+w,
\]
where
\[
w=\E_{\ybold}\left(\frac{1}{n}\E_{\ystar|\ybold}\|\ystar-H\ybold\|^2_2-\frac{1}{n}\|\ybold-H\ybold\|^2_2\right)
\]
and it is assumed that $\ystar$ and $\ybold$ are drawn from the same distribution and have the same predictor, $H\ybold.$ However, as was already discussed in the previous sections, these conditions are not satisfied in many use cases. The following prediction error generalizes $Loss(Opt),$
\[
Loss(Opt_{t})=\frac{1}{n}\|\ybold-H\ybold\|^2_2+w_{t},
\]
where
\[
w_{t}=\E_{\ybold}\left(\frac{1}{n^*}\E_{\ystar|\ybold}\|\ystar-H^*\ybold\|^2_2-\frac{1}{n}\|\ybold-H\ybold\|^2_2\right).
\]
Similarly to $tAIC$ definition, given a set of models $\mathcal{H},$ $Loss(Opt_{t})$ can be used for model selection as follows
\begin{align}
h_{best}=\underset{h\in\mathcal{H}}{\argmin}\,Loss_{h}(Opt_{t}),\label{Transductive Optimism criterion}
\end{align}
where $Loss_{h}(Opt_{t})$ is $Loss(Opt_{t})$ for model $h.$

Lemma \ref{Optimism lemma} introduces a general expression of $w_{t}$ for predictors that are linear in $\ybold.$
\begin{lemma} \label{Optimism lemma}
Let $\ybold\in\mathbb{R}^{n}$ be a random variable with mean $\mubold$ and variance $V.$ Similarly, let $\ystar\in\mathbb{R}^{n^*}$ be a random variable with mean $\mustar$ and variance $V^*.$ In addition, let $H\ybold\in R^{n}$ and $H^*\ybold\in R^{n^*}$ be the predictors of $\ybold$ and $\ystar$ respectively when $H,H^*$ don't contain $\ybold$ and $\ystar$.
Then 
\begin{align*}
w_{t} =&\frac{2}{n}tr\left(HV\right)-\frac{2}{n^*}tr\left(H^{*}\Cov(\ybold,\ystar)\right)\\
&+\frac{1}{n^*}tr\left(V^*\right)-\frac{1}{n}tr\left(V\right)+\frac{1}{n^*}tr\left(H^{*}VH^{*'}\right)-\frac{1}{n}tr\left(HVH^{'}\right)\\
&+\frac{1}{n}tr\left(2H\mubold\mutrbold-\mubold\mutrbold-H\mubold\mutrbold H^{'}\right)-\frac{1}{n^*}tr\left(2H^{*}\mubold\mustartr-\mustar\mustartr-H^{*}\mubold\mutrbold H^{*'}\right).
\end{align*}
\end{lemma}
\begin{corollary}
Given the definitions in lemma \ref{Optimism lemma}, when $H\mubold=\mubold$ and $H^*\mubold=\mustar$
\begin{align*}
w_{t}=&\frac{2}{n}tr\left(HV\right)-\frac{2}{n^*}tr\left(H^{*}\Cov(\ybold,\ystar)\right)\\
&+\frac{1}{n^*}tr\left(V^*\right)-\frac{1}{n}tr\left(V\right)+\frac{1}{n^*}tr\left(H^{*}VH^{*'}\right)-\frac{1}{n}tr\left(HVH^{'}\right).
\end{align*}
\end{corollary}
In case $H^*=H,$ $V^*=V$ and $V-\Cov(\ybold,\ystar)=\sigma^2I_n,$ then 
\[
w_{t}=w=\frac{2}{n}tr\left(H\right),
\]
which is the same result as was introduced by \cite{hodges2001counting} for Linear Hierarchical models.

$Loss(Opt_{t})$ is based on the squared error loss function which reflects Euclidean distance. Other prediction error estimators which are based on different distances, such as on Mahalanobis distance \citep{mahalanobis1936generalised}, might be suggested as well. Corollary \ref{corr Mahalanobis} presents a penalty correction for a prediction error estimator which is based on Mahalanobis distance.
\begin{corollary}\label{corr Mahalanobis}
Given the definitions in lemma \ref{Optimism lemma}, when $H\mubold=\mubold$ and $H^*\mubold=\mustar$
\begin{align}\label{Mahalanobis}
&\E_{\ybold}\left(\frac{1}{n^*}\E_{\ystar|\ybold}\|\ystar-H^*\ybold\|^2_M-\frac{1}{n}\|\ybold-H\ybold\|^2_M\right)\\\nonumber
&\qquad\qquad\qquad\qquad=\frac{2}{n}tr\left(R^{-1}HV\right)-\frac{2}{n^*}tr\left(R^{*^{-1}}H^{*}\Cov(\ybold,\ystar)\right)\\\nonumber
& \qquad\qquad\qquad\qquad\,\,\,\,\,\,\,+\frac{1}{n^*}tr\left(R^{*^{-1}}V^{*}\right)-\frac{1}{n}tr\left(R^{-1}V\right)\\\nonumber
 & \qquad\qquad\qquad\qquad\,\,\,\,\,\,\,+\frac{1}{n^*}tr\left(R^{*^{-1}}H^{*}VH^{*'}\right)-\frac{1}{n}tr\left(R^{-1}HVH^{'}\right)\\\nonumber
&\qquad\qquad\qquad\qquad=2C_{tAI}-\log\left(\frac{|R^{*}|^{\frac{1}{n^*}}}{|R|^{\frac{1}{n}}}\right),
\end{align}
where 
\begin{align*}
\|\ystar-H^*\ybold\|^2_M&=(\ystar-H^*\ybold)^{'}R^{*{-1}}(\ystar-H^*\ybold)\\
\|\ybold-H\ybold\|^2_M&=(\ybold-H\ybold)^{'}R^{-1}(\ybold-H\ybold).
\end{align*}
\end{corollary}
The relation between eq. (\ref{Mahalanobis}) and $C_{tAI}$ arises due to the relation between Mahalanobis distance and the normal likelihood which $tAI$ is based on.

It is natural to use $Loss(Opt_{t})$ instead of $tAI$ for linear predictors that don't assume normality, such as the predictors that are used in nearest neighbors, Nadaraya-Watson kernel regression and smoothing spline models. Moreover, due to the form of the normal density function, many predictors that seem to be based on the normality assumption can be alternatively interpreted as a solution of a least squares problem or complex versions of least squares problems like weighed least squares and penalized least squares problems. For example, GLS can be interpreted as the solution of weighted least squares problem with the weight matrix $V^{-1}$. Similarly, $\fhatstar$ can be interpreted as the solution of the following problem,
\[
\underset{\abold\in\mathbb{R}^{n^*} ,B\in\mathbb{R}^{n^*\times n}}{\mathrm{min}}\E_{\ystar,\ybold}\|\ystar-(\abold+B\ybold)\|^2_2.
\]
The proof is attached in in Appendix \ref{$tAIC$ derivation appendix}.
These alternative interpretations are free from normality assumption and therefore $Loss(Opt_{t})$ can be suitable for them. Since many predictors can be interpreted in different ways, then the assignation of predictors to $tAI$ or to $Loss(Opt_{t})$ should refer to the possibility to assume normality rather than to the predictor type.

\section{Discussion and Conclusions}
$tAI$ is an extension of the prediction error estimators that are used in $cAIC$ and $mAIC,$ extending them to estimate prediction error at interpolation and extrapolation points. As it is demonstrated in Section \ref{Use cases}, these use cases are common in various research fields, and particularly in Geostatistics and health, when GPR and LMM are used for predicting at interpolation and extrapolation points. Since GLS, linear regression and smoothing splines can be expressed as LMM \citep{brumback1999comment}, 
$tAI$ is applicable for them as well.

The correction in $tAI$ is more complicated than the corrections in $cAIC$ and $mAIC,$ which are $tr(H)/n$ and $p/n$ respectively. The correction in $tAI,$ is affected by the relations between $\Var(\ybold)$ to $\Var(\ystar),$ $\Var(\fhatbold)$ to $\Var(\fhatstar)$ and between $\Cov(\ybold,\ystar)$ to $\Var(\ybold).$ When interpreting the correction as a measure of over-fitting, the differences between the corrections gives a new perspective about how the over-fitting is composed as a function of the variance structure of the problem.

In many cases the variances parameters are unknown in advance and therefore are estimated by various procedures prior the model fitting, e.g. REML in LMM \citep{verbeke1997linear}. Estimating the variance parameters implies an extra variation for $tAI,$ especially when the sample size is small. Estimating the in-sample prediction error under the LMM setup when the variance parameters are unknown was addressed by \citep{liang2008note}. Extending this to a transductive setup is a challenge for a future work.

The numerical analyses emphasize the practical importance in using $tAI$ in scenarios where $\{X^*,Z^*,R^*\}\neq\{X,Z,R\}$ are different. It is noticeable
especially when predicting at extrapolation points, since in this case the differences between $\Var(\ybold)$ to $\Var(\ystar)$ and between $\Var(\fhatbold)$ to $\Var(\fhatstar)$ can be large.

$Loss(Opt_{t})$ is another prediction error estimator for cases involving predicting at interpolation and extrapolation points. Unlike $tAI,$ $Loss(Opt_{t})$ doesn't assume that the observations are normally distributed and therefore it is also applicable in various non-parametric applications. Since many predictors that are apparently based on normal linear model can be alternatively interpreted as solutions for the generalized least squares problems, the assignation of predictors to $tAI$ or to $Loss(Opt_{t})$ should refer to the possibility to assume normality rather than to the predictor formula.

\appendix
\section{Proofs}\label{$tAIC$ derivation appendix}
\begin{proof}[Proof of Theorem \ref{Main Theorem}]
\begin{align*}
C_{tAI}=&\E_{\ybold}\left[-\frac{1}{n^*}\E_{\ystar|\ybold}l(\ystar)-\left\{-\frac{1}{n}l(\ybold)\right\}\right]\\
=&-\frac{1}{n^*}\E_{\ybold}\E_{\ystar}l(\ystar)+\frac{1}{n}\E_{\ybold}l(\ybold)\\
=&\frac{1}{2n^*}\left\{\log|R^{*}|+n^*\log(2\pi)+\E_{\ybold}\E_{\ystar}(\ystar-H^*\ybold)^{'}R^{*^{-1}}(\ystar-H^*\ybold)\right\}\\
 &-\frac{1}{2n}\left\{\log|R|+n\log(2\pi)+\E_{\ybold}(\ybold-H\ybold)^{'}R^{-1}(\ybold-H\ybold)\right\}\\
 =&\frac{1}{2}\log\left(\frac{|R^{*}|^{\frac{1}{n^*}}}{|R|^{\frac{1}{n}}}\right)\\
 &+\frac{1}{2}\left\{\frac{1}{n^*}\E_{\ybold}\E_{\ystar}(\ystar-H^{*}\ybold)^{'}R^{*^{-1}}(\ystar-H^{*}\ybold)-\frac{1}{n}\E_{\ybold}(\ybold-H\ybold)^{'}R^{-1}(\ybold-H\ybold)\right\}
\end{align*}

Since
\begin{align*}
&\E_{\ybold}\E_{\ystar}(\ystar-H^{*}\ybold)^{'}R^{*^{-1}}(\ystar-H^{*}\ybold)\\
&\qquad\qquad\qquad=tr\left\{R^{*^{-1}}\E_{\ystar}(\ystar \ytrstar)\right\}+tr\left\{H^{*'}R^{*^{-1}}H^{*}\E_{\ybold}(\ybold\ytrbold )\right\}\\ 
 & \qquad\qquad\qquad\,\,\,\,\,\,\,-2tr\left\{R^{*^{-1}}H^{*}\E_{\ybold}\E_{\ystar}(\ybold\ytrstar)\right\}\\
&\qquad\qquad\qquad=tr\left\{R^{*^{-1}}(V^{*}+\mustar\mustartr)\right\}+tr\left\{R^{*^{-1}}H^{*}(V+\mubold\mutrbold)H^{*'}\right\}\\ 
 & \qquad\qquad\qquad\,\,\,\,\,\,\,-2tr\left\{R^{*^{-1}}H^{*}(\Cov(\ybold,\ystar)+\mubold\mustartr)\right\}
\end{align*}
and
\begin{align*}
\E_{\ybold}(\ybold-H\ybold)^{'}R^{-1}(\ybold-H\ybold) =&tr\left\{R^{-1}\E_{\ybold}(\ybold\ytrbold )\right\}+tr\left\{(H^{'}R^{-1}H\E_{\ybold}(\ybold\ytrbold )\right\}\\
&-2tr\left\{R^{-1}H\E_{\ybold}(\ybold\ytrbold )\right\}\\
=&tr\left\{R^{-1}(V+\mubold\mutrbold)\right\}+tr\left\{R^{-1}H(V+\mubold\mutrbold)H^{'}\right\}
\\
&-2tr\left\{R^{-1}H(V+\mubold\mutrbold)\right\}
\end{align*}
then
\begin{align*}
 &\frac{1}{n^*}\E_{\ybold}\E_{\ystar}(\ystar-H^{*}\ybold)^{'}R^{*^{-1}}(\ystar-H^{*}\ybold)-\frac{1}{n}\E_{\ybold}(\ybold-H\ybold)^{'}R^{-1}(\ybold-H\ybold)\\
  &=\frac{2}{n}tr\left(R^{-1}HV\right)-\frac{2}{n^*}tr\left(R^{*^{-1}}H^{*}\Cov(\ybold,\ystar)\right)\\
 &\,\,\,\,\,\,\,+\frac{1}{n^*}tr\left(R^{*^{-1}}V^{*}\right)-\frac{1}{n}tr\left(R^{-1}V\right)+\frac{1}{n^*}tr\left(R^{*^{-1}}H^{*}VH^{*'}\right)-\frac{1}{n}tr\left(R^{-1}HVH^{'}\right)\\
 &\,\,\,\,\,\,\,+\frac{1}{n}tr\left(R^{-1}(2H\mubold\mutrbold-\mubold\mutrbold-H\mubold\mutrbold H^{'})\right)\\
&\,\,\,\,\,\,\, -\frac{1}{n^*}tr\left(R^{*^{-1}}(2H^{*}\mubold\mustartr-\mustar\mustartr-H^{*}\mubold\mutrbold H^{*'})\right)
\end{align*}

Since $H\mubold=\mubold$ 
\begin{align*}
H\mubold\mutrbold & =\mubold\mutrbold\\
H\mubold\mutrbold H^{'} & =\mubold\mutrbold
\end{align*}
and similarly since $H^{*}\mubold=\mustar$ 
\begin{align*}
H^{*}\mubold\mustartr & =\mustar\mustartr\\
H^{*}\mubold\mutrbold H^{*'} & =\mustar\mustartr
\end{align*}
then
\begin{align*}
\frac{1}{n}tr\left(R^{-1}(2H\mubold\mutrbold-\mubold\mutrbold-H\mubold\mutrbold H^{'})\right)&=0\\
\frac{1}{n^*}tr\left(R^{*^{-1}}(2H^{*}\mubold\mustartr-\mustar\mustartr-H^{*}\mubold\mutrbold H^{*'})\right)&=0
\end{align*}
and therefore
\begin{align*}
 & \frac{1}{n^*}\E_{\ybold}\E_{\ystar}(\ystar-H^{*}\ybold)^{'}R^{*^{-1}}(\ystar-H^{*}\ybold)-\frac{1}{n}\E_{\ybold}(\ybold-H\ybold)^{'}R^{-1}(\ybold-H\ybold))\\
 & =\frac{2}{n}tr\left(R^{-1}HV\right)-\frac{2}{n^*}tr\left(R^{*^{-1}}H^{*}\Cov(\ybold,\ystar)\right)\\
 &\,\,\,\,\,\,\,+\frac{1}{n^*}tr\left(R^{*^{-1}}V^{*}\right)-\frac{1}{n}tr\left(R^{-1}V\right)+\frac{1}{n^*}tr\left(R^{*^{-1}}H^{*}VH^{*'}\right)-\frac{1}{n}tr\left(R^{-1}HVH^{'}\right)
\end{align*}
which gives
\begin{align*}
C_{tAI} =&\frac{1}{n}tr\left(R^{-1}HV\right)-\frac{1}{n^*}tr\left(R^{*^{-1}}H^{*}\Cov(\ybold,\ystar)\right)\\
 & +\frac{1}{2}\left\{\log\left(\frac{|R^{*}|^{\frac{1}{n^*}}}{|R|^{\frac{1}{n}}}\right)+\frac{1}{n^*}tr\left(R^{*^{-1}}V^{*}\right)-\frac{1}{n}tr\left(R^{-1}V\right)\right\}\\
 & +\frac{1}{2}\left\{\frac{1}{n^*}tr\left(R^{*^{-1}}H^{*}VH^{*'}\right)-\frac{1}{n}tr\left(R^{-1}HVH^{'}\right)\right\}
\end{align*}
\end{proof}

\begin{proof}[Proof of Lemma \ref{Optimism lemma}]
Under some regularity conditions
\begin{align*}
\frac{\partial \E_{\ybold}\E_{\ystar}\lVert \ystar-(\abold+B\ybold)\rVert_{2}^{2}}{\partial \abold} & =\frac{\E_{\ybold}\E_{\ystar}\partial\lVert \ystar-(\abold+B\ybold)\rVert_{2}^{2}}{\partial \abold}\\
\frac{\partial \E_{\ybold}\E_{\ystar}\lVert \ystar-(\abold+B\ybold)\rVert_{2}^{2}}{\partial B} & =\frac{\E_{\ybold}\E_{\ystar}\partial\lVert \ystar-(\abold+B\ybold)\rVert_{2}^{2}}{\partial B}.
\end{align*}

Since 
\begin{align*}
\frac{\partial\rVert \ystar-(\abold+B\ybold)\lVert_{2}^{2}}{\partial \abold} & =\frac{\partial\left(-\ytrstar \abold-\atrbold\ystar+\atrbold\abold+\atrbold B\ybold+\ytrbold B^{'}\abold\right)}{\partial \abold}\\
 & =-2\ystar+2\abold+2B\ybold,
\end{align*}
then
\[
\frac{\partial \E_{\ybold}\E_{\ystar}\rVert \ystar-(\abold+B\ybold)\lVert_{2}^{2}}{\partial \abold}=-2\mustar+2\abold+2B\mu.
\]

Similarly 
\begin{align*}
\frac{\partial\rVert \ystar-(\abold+B\ybold)\lVert_{2}^{2}}{\partial B} & =\frac{\partial\left(-\ytrstar B\ybold+\atrbold B\ybold-\ytrbold B^{'}\ystar+\ytrbold B^{'}\abold+\ytrbold B^{'}B\ybold\right)}{\partial B}\\
 & =-2\ystar \ytrbold +2\abold\ytrbold +2B\ybold\ytrbold 
\end{align*}
and therefore
\[
\frac{\partial \E_{\ybold}\E_{\ystar}\rVert \ystar-(\abold+B\ybold)\lVert_{2}^{2}}{\partial B}=-2\left(\Cov(\ystar,\ybold)+\mustartr\mubold\right)+2\abold\mutrbold+2B\left(V+\mubold\mutrbold\right).
\]

Since the optimized function is convex, the solution of the following equations achieves the global minimum where
\begin{align*}
0=&-\mustar+\abold+B\mubold\\
0=&-\Cov(\ystar,\ybold)-\mustar\mutrbold+\abold\mutrbold+B\left(V+\mubold\mutrbold\right).
\end{align*}

The solution for $B$ is: 
\begin{align*}
B\left(V+\mubold\mutrbold\right)& =\Cov(\ystar,\ybold)+\mustar\mutrbold-\abold\mutrbold\\
 & =\Cov(\ystar,\ybold)+\mustar\mutrbold+\left(\mustar-B\mu\right)\mutrbold\\
 & =\Cov(\ystar,\ybold)+B\mubold\mutrbold,
\end{align*}
which gives
\[
B=\Cov(\ystar,\ybold)V^{-1}.
\]
The solution for $\abold$ is
\begin{align*}
\abold & =\mustar-B\mubold\\
 & =\mustar-\Cov(\ystar,\ybold)V^{-1}\mubold.
\end{align*}
Therefore the optimal linear equation is the same as $\fhatstar.$
\end{proof}

\section{Scenarios in  mixed model where $R\neq \sigma^2_{\epsilon}I$}\label{GLS appendix}
\begin{example}
Consider the following model
\[
\ybold =X\betabold+Z_{1}\bonebold+Z_{2}\btwobold+\epsilonbold,
\]
where $\epsilonbold\sim N(0,\sigma^{2}_{\epsilon}I),\,\bonebold\sim N(0,G_{1}),\,\btwobold\sim N(0,G_{2}),$
\[
\ystar =X\betabold+Z_{1}\bonebold+Z_{2}^{*}\btwostarbold+\epsilonstarbold
\]
where $\epsilonstarbold\sim N(0,\sigma^{2}_{\epsilon}I),\,\btwostarbold\sim N(0,G_{2}),$ 
\[
\epsilonbold\perp\epsilonstarbold\perp \bonebold\perp \btwobold\perp \btwostarbold,
\]
and $Z_{1}\in\mathbb{R}^{nXq_{1}},\,Z_{2}\in\mathbb{R}^{nXq_{2}},\,Z_{2}^*\in\mathbb{R}^{nXq_{2}}.$

Since $\bonebold$ is common for $\ybold$ and $\ystar,$ its estimate can be utilized for achieving a better accuracy in predicting $\ystar.$ Since $\ystar$ doesn't contain $\btwobold$ and $\epsilonbold$, estimating them doesn't contribute achieving a better accuracy in predicting $\ystar.$ Therefore, in terms of predicting $\ystar$ using $\fhatstar$, the following model definition has the same predicting formula as the previous one, 
\[
\ybold  =X\betabold+Z_{1}\bonebold+\epsilonbold,
\]
where $\epsilonbold\sim N(0,\sigma^{2}_{\epsilon}I+Z_{2}G_{2}Z_{2}^{'}),\,\bonebold\sim N(0,G_{1}),$
\[
\ystar =X\betabold+Z_{1}\bonebold+\epsilonstarbold,
\]
where $\epsilonstarbold\sim N(0,\sigma^{2}_{\epsilon}I+Z_{2}^{*}G_{2}Z_{2}^{*'}),$ 
\[
\epsilonbold\perp\epsilonstarbold\perp \bonebold.
\]
Since the second formulation is simpler it can be preferred when the goal is predicting $\ystar.$
\end{example}
\begin{example} Consider the standard LMM setup when $\ybold\in\mathbb{R}^n$ is drawn from $K$ clusters, however each observation, $y_i,$ is an average of $w_{i}$ i.i.d observations, $y_i=(\sum_{l=1}^{w_i}\phi_{i,l})/w_i,$ where $\phi_{i,l}\sim N(0,\sigma^{2}_{\epsilon}).$ Assume $\phi_{i,l}$ are unknown, however $w_i$ is known. The variance of the residual, in this case is 
\[
\left(\begin{array}{cccc}
\frac{\sigma^{2}_{\epsilon}}{w_{1}} & 0 & \cdots & 0\\
0 & \frac{\sigma^{2}_{\epsilon}}{w_{2}} &  & \vdots\\
\vdots &  & \ddots & 0\\
0 & \cdots & 0 & \frac{\sigma^{2}_{\epsilon}}{w_{n}}
\end{array}\right).
\]
\end{example}

Another common use case is when due to poor available data, technical restrictions or other reasons, part of the correlation of $\ybold$ is not explained by the random effects. In those cases, this part will be expressed by the residual, $\epsilonbold,$ and therefore $\Var(\epsilonbold)$ will be a non-diagonal matrix.

\section{Additional Numerical Results}\label{Additional Numeriocal Results}
Figures \ref{Density1} and \ref{Density2} present the distributions of $tAI,\,cAI,\,mAI$ and $-\E_{\ystar|\ybold}l(\ystar)/n^{*}$ for models $1$ and $2$. For more details, see Section \ref{Numerical part}.

\begin{figure}[h!]
 \centering
    \begin{subfigure}[b]{0.5\textwidth}
        \centering
      \includegraphics[width=1\linewidth]{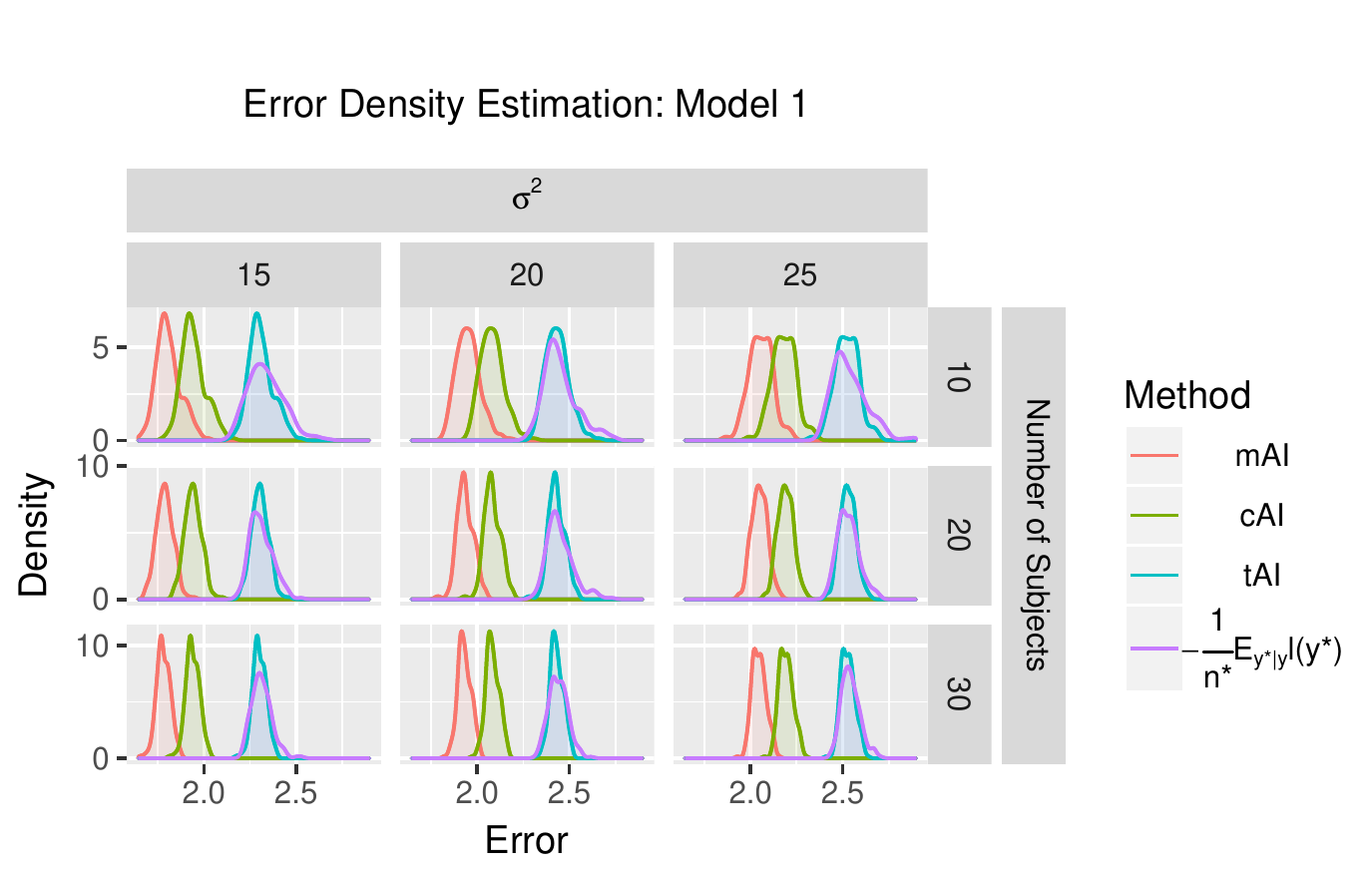}
          \caption{Density: Model 1}
            \label{Density1}
            \end{subfigure}%
    ~ 
    \begin{subfigure}[b]{0.5\textwidth}
        \centering
          \includegraphics[width=1\linewidth]{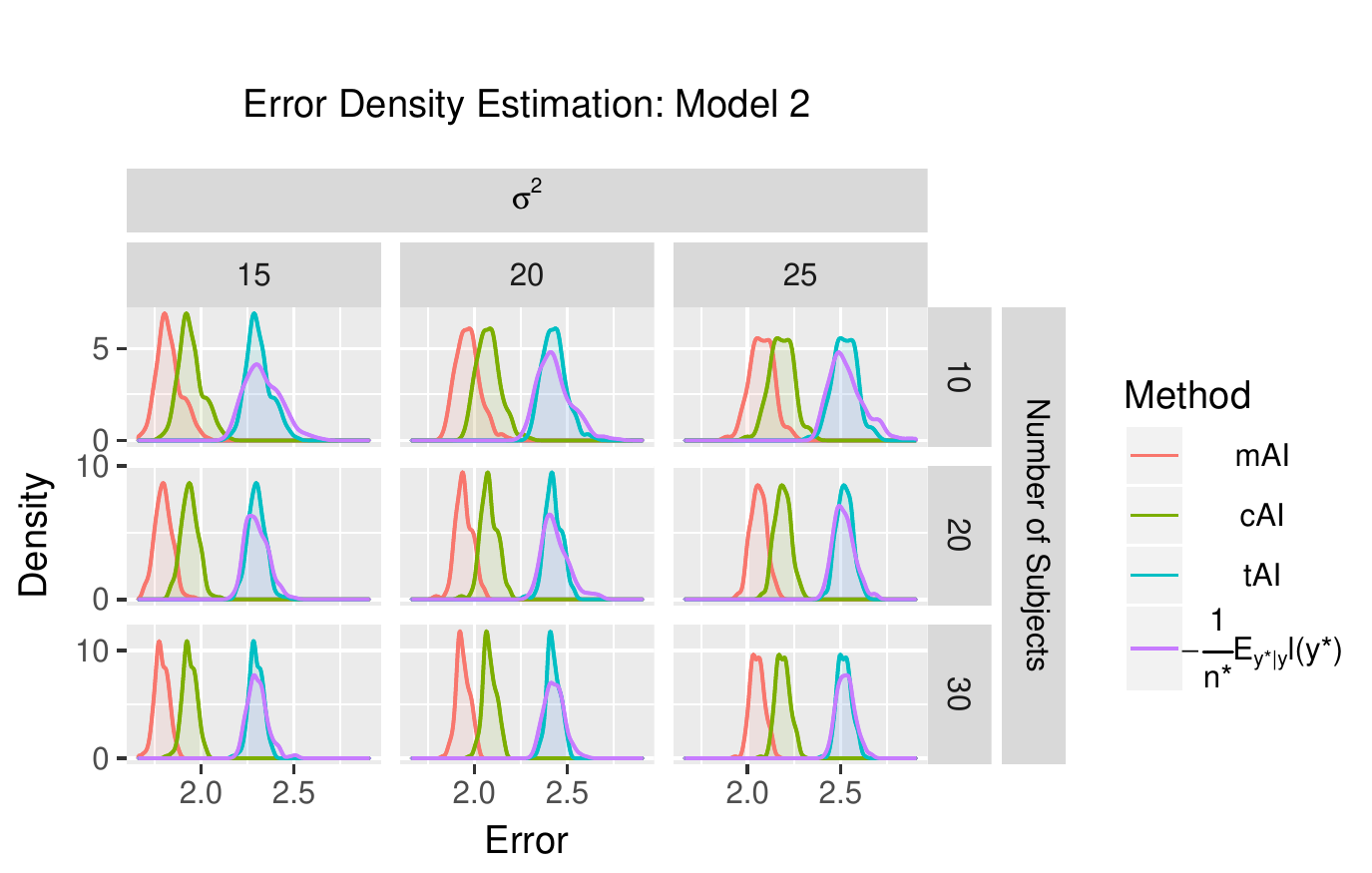}
          \caption{Density: Model 2}
            \label{Density2}
      \end{subfigure}
    \caption{Densities of  $tAI,\,cAI,\,mAI$ and $-\E_{\ystar|\ybold}l(\ystar)/n^{*}$ as a function of the sample size and $\sigma^2$.}
    \label{Density12}
\end{figure}

Figure \ref{Average Argmin EE} presents the error 
\[
\E_{\ybold}\E_{\ystar|\ybold}-\frac{1}{n^{*}}l_{h_{best}}(\ystar)
\]
for each one of the model selection criteria, $tAIC,\,cAIC\,mAIC$ and the oracle criterion
\begin{align*}
h_{best}&=\underset{h\in\{1,2,3\}}{\argmin}-\frac{1}{n^{*}}\E_{\ybold}\E_{\ystar|\ybold}l_h(\ystar)
\end{align*}
in the nine setups. For more details see Section \ref{Numerical part}.
\begin{figure}[h!]
\begin{centering}
\includegraphics[width=1\linewidth]{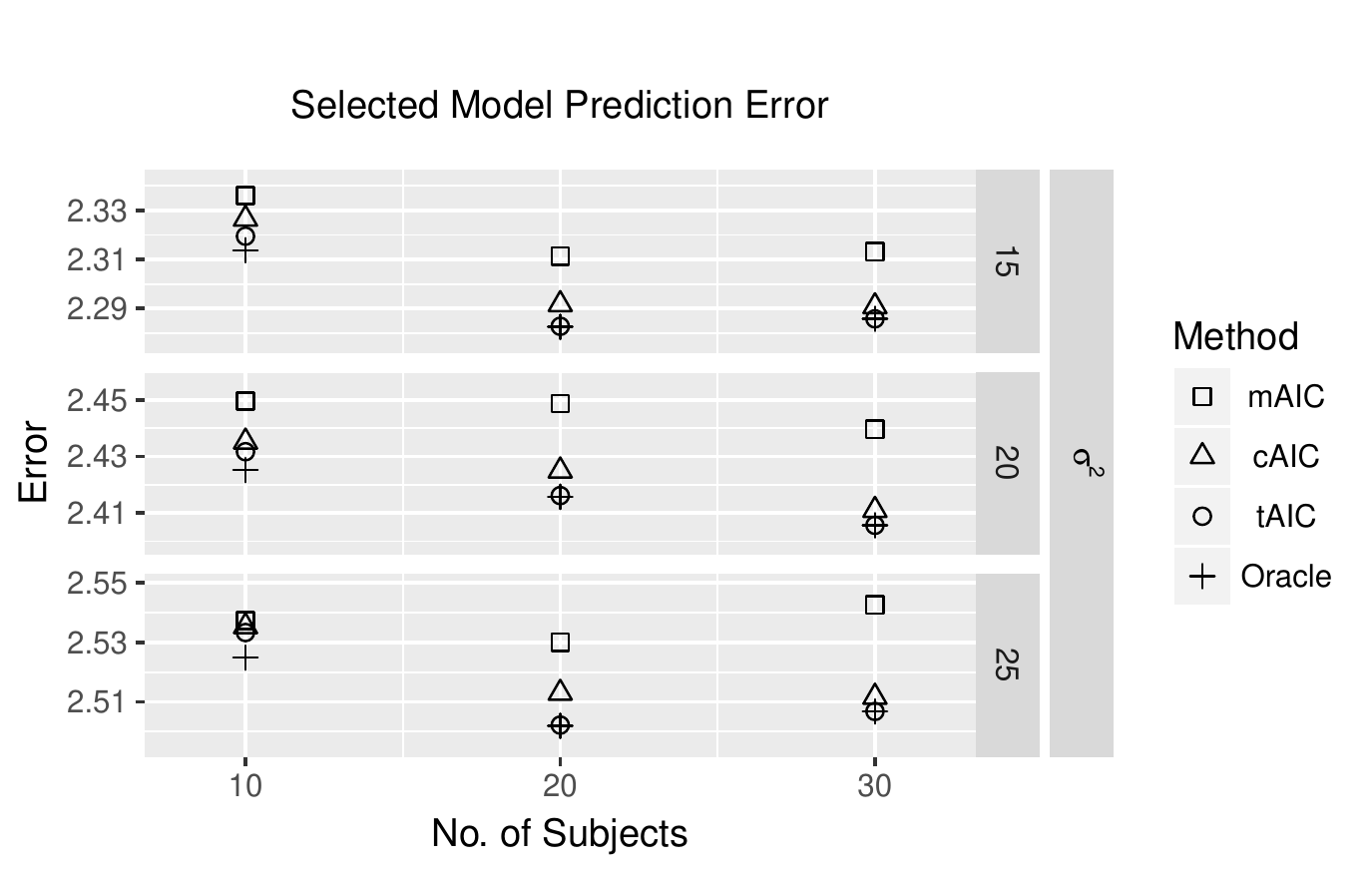}
\caption{For each setup, each symbol refers to the prediction error $\E_{\ystar|\ybold}-l_{h_{best}}(\ystar)/n^{*}$ of the relevant criterion, $mAIC,\,cAIC\,tAIC$ and the oracle criterion.}
\label{Average Argmin EE}
\end{centering}
\end{figure}

Figure \ref{Agreement Rate EE} presents the agreement rate of the criteria, $tAIC,\,cAIC$ and $mAIC$ with the oracle criterion 
\begin{align*}
h_{best}&=\underset{h\in\{1,2,3\}}{\argmin}-\frac{1}{n^{*}}\E_{\ybold}\E_{\ystar|\ybold}l_h(\ystar).
\end{align*}
For more details see Section \ref{Numerical part}.
\begin{figure}[h!]
\begin{centering}
\includegraphics[width=1\linewidth]{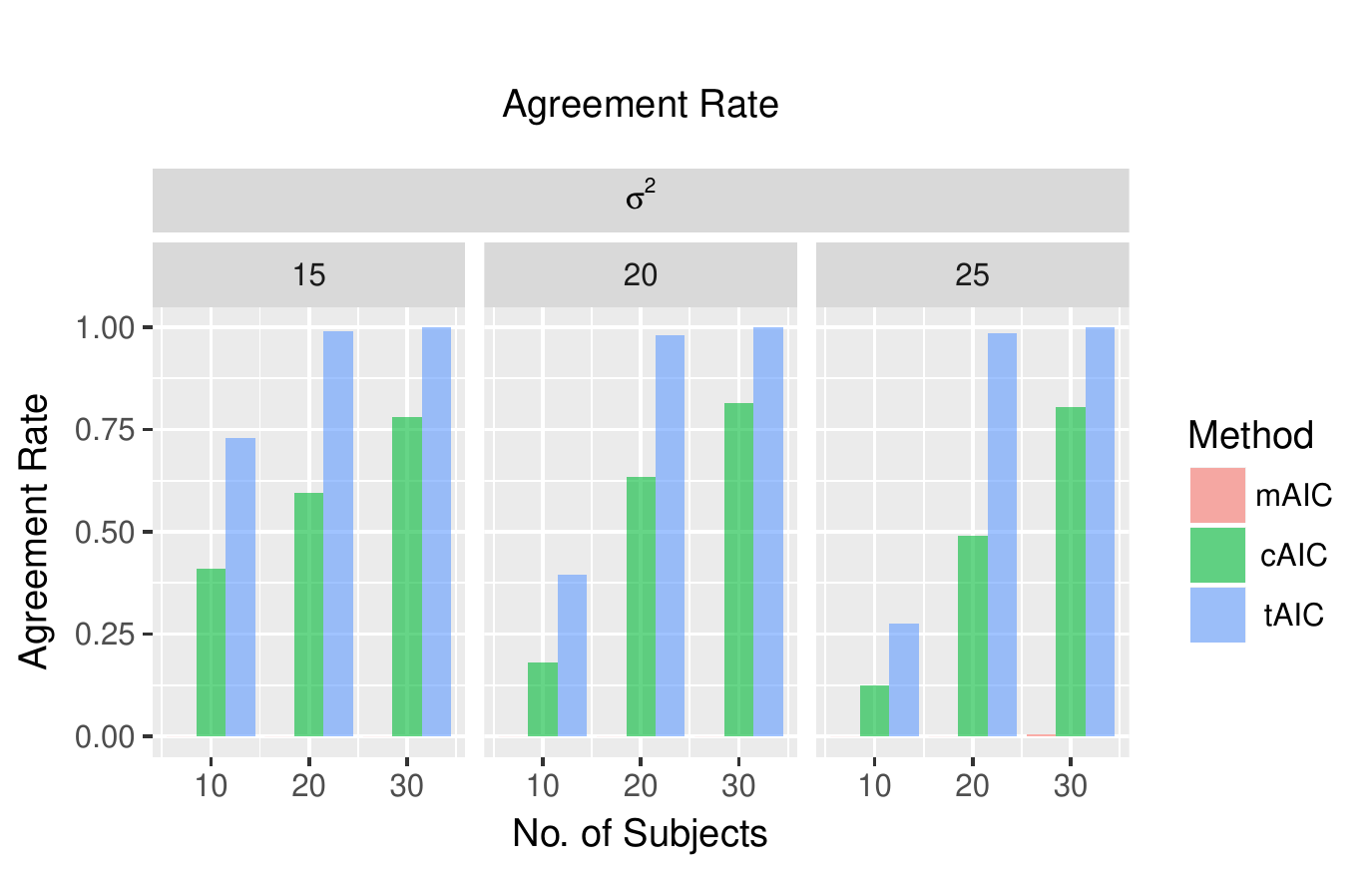}
\caption{For each setup, each bar refers to the agreement rate of the relevant criterion with the oracle criterion}
\label{Agreement Rate EE}
\end{centering}
\end{figure}

\bibliography{bibliography.bib}
\bibliographystyle{Chicago}



\end{document}